\documentclass[11pt]{article}

\usepackage{lmodern}

\usepackage[margin=1in]{geometry}
\usepackage{amsfonts,amssymb,amsmath,amsthm,mathtools}
\usepackage{thmtools}

\usepackage{graphicx,latexsym,lpic,bm,xspace,booktabs,array,multirow,changepage}

\usepackage{microtype}

\usepackage[bf]{caption}

\usepackage[usenames,dvipsnames]{xcolor}

\usepackage[bottom]{footmisc}
\usepackage{todonotes}

\usepackage{hyperref}

\usepackage{algorithmicx,algorithm}
\usepackage[noend]{algpseudocode}
\algnewcommand\algorithmicinput{\text{Input:}}
\algnewcommand\algorithmicoutput{\text{Output:}}
\algnewcommand\Input{\item[\algorithmicinput]}
\algnewcommand\Output{\item[\algorithmicoutput]}
\algrenewcommand\algorithmicindent{1.5em}

\usepackage{dsfont}
\usepackage[greek,english]{babel}
\addto\extrasenglish{}
\addto\extrasenglish{}

\usepackage[nottoc]{tocbibind}
\usepackage{tocloft}

\usepackage{enumitem}
\setlist[itemize]{noitemsep,label=$-$}
\setlist[enumerate]{noitemsep}

\hypersetup{
	colorlinks=true,
	linkcolor=Sepia,
	citecolor=Sepia,
	filecolor=Sepia,
	urlcolor=Sepia
}

\usepackage{tikz}
\usepackage[framemethod=tikz]{mdframed}
\mdfsetup{
	roundcorner=4pt,
	nobreak=true,
	linewidth=.5,
	innerleftmargin=10pt,
	innerrightmargin=10pt,
	innertopmargin=10pt,
	innerbottommargin=10pt,
	skipabove=12,skipbelow=5pt
}
\definecolor{myGold}{RGB}{231,141,20}
\definecolor{myBlue}{rgb}{0.19,0.41,.65}
\definecolor{myPurple}{RGB}{175,0,124}
\definecolor{myGreen}{RGB}{14,126,6}

\usetikzlibrary{arrows.meta,positioning} 

\declaretheorem[name=Theorem]{theorem}
\declaretheorem[name=Lemma,sibling=theorem]{lemma}
\declaretheorem[name=Claim,sibling=theorem]{claim}
\declaretheorem[name=Fact,sibling=theorem]{fact}

\declaretheorem[name=Definition,style=definition]{definition}
\declaretheorem[name=Conjecture,style=definition]{conjecture}

\declaretheorem[name={Simplex\hspace{.4em}Lemma},numbered=no]{sim-lemma}
\newcommand{\simlemma}{\hyperref[lem:sim-lemma]{Simplex Lemma}\xspace}

\declaretheorem[name={Round~Lemma},numbered=no]{round-lemma}
\newcommand{\roundlemma}{\hyperref[lem:round-lemma]{Round Lemma}\xspace}

\declaretheorem[name={Large~Index~Lemma},numbered=no]{index-lemma}
\newcommand{\indexlemma}{\hyperref[lem:index-lemma]{Large Index Lemma}\xspace}

\DeclareMathOperator{\poly}{poly}

\DeclareMathOperator{\Cube}{Cube}
\DeclareMathOperator{\free}{free}
\DeclareMathOperator{\fix}{fix}
\DeclareMathOperator{\err}{err}

\DeclareMathOperator{\myRef}{Ref}

\newcommand{\Alice}{\text{\slshape Alice}}
\newcommand{\Bob}{\text{\slshape Bob}}

\newcommand{\smallFunction}[2]{\newcommand{#1}{{\textsc{#2}}}}
\smallFunction{\Ind}{Ind}
\smallFunction{\PMInd}{PMInd}
\smallFunction{\Xor}{Xor}
\smallFunction{\GHS}{Gap-Hitting-Set}

\newcommand{\newsftext}[2]{\newcommand{#1}{{\text{\upshape\sffamily #2}}}}
\newsftext{\cc}{cc}
\newsftext{\dt}{dt}

\newcommand{\newclass}[2]{\newcommand{#1}{{\text{\upshape\sffamily #2}}\xspace}}
\renewcommand{\P}{{\text{\upshape\sffamily P}}\xspace}
\renewcommand{\d}{{\text{\upshape\sffamily d}}\xspace}
\newclass{\NP}{NP}
\newclass{\coNP}{coNP}

\newclass{\rectDag}{rect-dag}
\newclass{\decDag}{dec-dag}
\newclass{\ltfDag}{ltf-dag}
\newclass{\simDag}{sim-dag}

\newclass{\ltfTree}{ltf-tree}
\newclass{\decTree}{dec-tree}

\newclass{\res}{res}
\newclass{\cut}{cut}

\newclass{\resTree}{res-tree}
\newclass{\cutTree}{cut-tree}

\newclass{\w}{w}
\newclass{\bw}{bw}

\newclass{\rcc}{rCC}

\DeclareMathOperator*{\Exp}{\mathbf{E}}

\newcommand{\set}[1]{\left\{ #1 \right\}}

\newcommand{\R}{\mathbb{R}}

\newcommand{\calA}{\mathcal{A}}
\newcommand{\calB}{\mathcal{B}}

\newcommand{\calF}{\mathcal{F}}
\newcommand{\calI}{\mathcal{I}}
\newcommand{\calP}{\mathcal{P}}
\newcommand{\calO}{\mathcal{O}}
\newcommand{\calT}{\mathcal{T}}

\newcommand{\calS}{\mathcal{S}}

\newcommand{\calX}{\mathcal{X}}
\newcommand{\calY}{\mathcal{Y}}

\renewcommand{\Pr}{\mathbf{Pr}}
\newcommand{\Ent}{\mathbf{H}}
\newcommand{\Hmin}{\mathbf{H}_\infty}
\newcommand{\Dmin}{\mathbf{D}_\infty}

\newcommand{\x}{\bm{x}}
\newcommand{\y}{\bm{y}}

\newcommand{\X}{\bm{X}}
\newcommand{\Y}{\bm{Y}}

\usepackage[toc]{multitoc}

\begin{document}

\newgeometry{margin=1in,top=1.6in,bottom=1in}

\begin{center}
{\LARGE Automating Cutting Planes is $\NP$-Hard}
\\[10mm]

\large
\setlength\tabcolsep{0em}
\newcommand{\myPad}{\hspace{2.5em}}
\begin{tabular}{c@{\myPad}c@{\myPad}c@{\myPad}c}
Mika G\"o\"os$^\dagger$ &
Sajin Koroth &
Ian Mertz &
Toniann Pitassi\\[-.5mm]
\small\slshape Stanford University&
\small\slshape Simon Fraser University &
\small\slshape Uni.\ of Toronto&
\small\slshape Uni.\ of Toronto \& IAS
\end{tabular}

\vspace{8mm}

\large\today

\vspace{8mm}

\normalsize
\bf Abstract
\end{center}
\begin{adjustwidth}{7.5mm}{7.5mm}
We show that Cutting Planes (CP) proofs are hard to find: Given an unsatisfiable formula~$F$,
\begin{enumerate}[label=(\arabic*),itemsep=1mm]
	\item it is $\NP$-hard to find a CP refutation of $F$ in time polynomial in the length of the shortest such refutation; and
	\item unless {\scshape Gap-Hitting-Set} admits a nontrivial algorithm, one cannot find a \emph{tree-like} CP refutation of $F$ in time polynomial in the length of the shortest such refutation.
\end{enumerate}
The first result extends the recent breakthrough of Atserias and M\"uller ({\small FOCS 2019}) that established an analogous result for Resolution. Our proofs rely on two new lifting theorems: (1)~Dag-like lifting for gadgets with \emph{many output bits}. (2)~Tree-like lifting that simulates an $r$-round protocol with gadgets of query complexity $O(\log r)$ independent of input length.
\end{adjustwidth}

\vspace{4mm}
\setlength{\cftbeforesecskip}{8pt}
\setcounter{tocdepth}{2}
\tableofcontents

\renewcommand*{\thefootnote}{\fnsymbol{footnote}}
\footnotetext[2]{Part of the work done while at Institute for Advanced Study.}
\renewcommand*{\thefootnote}{\arabic{footnote}}
\setcounter{footnote}{0}

\thispagestyle{empty}
\setcounter{page}{0}
\newpage
\restoregeometry

\section{Introduction}

Propositional proof systems are by nature \emph{non-deterministic}: a short refutation of a formula~$F$ in a particular proof system constitutes an easy-to-check certificate (an $\NP$-witness) of~$F$'s unsatisfiability (which is a $\coNP$-property). The question of efficiently finding such refutations is the foundational problem of \emph{automated theorem proving} with applications to algorithm design, e.g., for combinatorial optimization~\cite{Fleming2019}. The following definition is due to Bonet et al.~\cite{bonet00interpolation}.

\begin{quote}
{\bf\itshape Automatability.}~
A proof system $\calP$ is \emph{automatable} if there is an algorithm that on input an unsatisfiable CNF formula $F$ outputs some $\calP$-refutation of $F$ in time polynomial in the length (or size) of the shortest $\calP$-refutation of $F$.
\end{quote}

\emph{Algorithms.}
Several basic propositional proof systems are automatable when restricted to proofs of bounded \emph{width} or \emph{degree}. For example, Resolution refutations of width $w$ can be found in time $n^{O(w)}$ for $n$-variate formulas~\cite{ben-sasson01short}. Efficient algorithms also exist for finding bounded-degree refutations in algebraic proof systems such as Nullstellensatz, Polynomial Calculus~\cite{cei}, Sherali--Adams, and Sum-of-Squares (under technical assumptions)~\cite{ODonnell2017,RaghavendraW17}.

\emph{Hardness.}
Without restrictions on width or degree, many of these systems are known \emph{not} to be automatable. For the most basic system, Resolution, a long line of work~\cite{Iwama1997,Alekhnovich2001,Alekhnovich2008,Mertz2019} recently culminated in an optimal non-automatability result by Atserias and M\"uller~\cite{Atserias2019}. They showed that Resolution is not automatable unless $\P=\NP$. Under stronger hardness assumptions non-automatability results are known for Nullstellensatz and Polynomial Calculus~\cite{Galesi2010,Mertz2019} as well as for various Frege systems~\cite{KrajicekP98,BonetPR97,BonetDGMP04}.

\paragraph{This work.}
The above list conspicuously omits to mention any hardness results for the Cutting Planes (CP) proof system (defined in \autoref{sec:cp} below). Indeed, we show the first such results:
\begin{itemize}[leftmargin=15mm,itemsep=1mm]
	\item[(\S\ref{sec:dag-result})] It is $\NP$-hard to automate CP. This is an Atserias--M\"uller style result for CP.
	\item[(\S\ref{sec:tree-result})] Under a stronger assumption, it is hard to automate \emph{tree-like} CP.
\end{itemize}
One reason Cutting Planes has been lacking non-automatability results is because of the shortage of techniques to prove lower bounds on CP refutation length. Virtually the only known method has been to find reductions to \emph{monotone circuit} lower bounds (for example, via monotone feasible interpolation). Our proofs rely on two new \emph{lifting theorems}, one of which bypasses the need for monotone circuit lower bounds. See \autoref{sec:overview} for an overview of our techniques.

\subsection{Cutting Planes} \label{sec:cp}

Cook, Coullard, and Tur\'an~\cite{cook87complexity} introduced Cutting Planes as a propositional proof system inspired by a like-named method to solve integer linear programs. The method uses rounding of linear inequalities (Chv\'atal--Gomory cuts) to reason about the integral solutions to a linear program.

The proof system version of CP is defined as follows. Suppose we are given a CNF formula $F$ over variables $x_1,\ldots,x_n$. A (dag-like) \emph{Cutting Planes refutation} of $F$ is a sequence of \emph{lines} $\ell_1,\ldots,\ell_m$ (where $m$ is the \emph{length}), each line being a linear inequality, $\sum_i a_ix_i\geq b$, with integer coefficients, $a_i,b\in\mathbb{Z}$. We require that the sequence ends with the contradictory inequality $\ell_m\coloneqq [\,0\geq 1\,]$ and that each $\ell_i$ satisfies one of the following:
\begin{itemize}[itemsep=2mm]
\item \emph{Axiom.}~ Line $\ell_i$ is either a boolean axiom ($x_i\geq 0$ or $-x_i\geq -1$) or an encoding of a clause of~$F$ (for example, clause $(x_1\lor \bar{x}_2)$ gets encoded as $x_1+(1-x_2)\geq 1$).
\item \emph{Derivation.}~ Line $\ell_i$ is deduced from two \emph{premises} $\ell_j, \ell_{j'}$ where $j,j'<i$ (perhaps $j=j'$) by an application of a \emph{sound rule}. (A refutation is \emph{tree-like} if each line appears at most once as a premise.)
\end{itemize}
In the original paper~\cite{cook87complexity} the rules were: (1) deriving from $\ell_j,\ell_{j'}$ any nonnegative integer linear combination of them, and (2) deriving from $\sum a_ix_i\geq b$ the line $\sum (a_i/c)x_i\geq \lceil b/c\rceil$ where $c\coloneqq \operatorname{gcd}(a_1,\ldots,a_n)$. Stronger rules have also been studied, e.g.,~\cite{Cook1990,Balas1993}, the most general being the \emph{semantic rule}, which allows \emph{any} sound inference: $\ell_i$ can be derived from $\ell_j,\ell_{j'}$ provided every boolean vector $x\in\{0,1\}^n$ that satisfies both $\ell_j$ and $\ell_{j'}$ also satisfies $\ell_i$. In this paper, we adopt the best of all possible worlds: our lower bounds on CP refutation length will hold even against the semantic system and our upper bounds use the weakest possible rules (in fact, our upper bounds hold for Resolution, which is simulated by every variety of CP).

\subsection{Dag-like result} \label{sec:dag-result}

Our first main result is a CP analogue of the Atserias--M\"uller theorem~\cite{Atserias2019}.
\begin{theorem}[Dag-like] \label{thm:main-dag}
There is a polynomial-time algorithm $\calA$ that on input an $n$-variate 3-CNF formula $F$ outputs an unsatisfiable CNF formula $\calA(F)$ such that:
\begin{itemize}
\item If $F$ is \emph{satisfiable}, then $\calA(F)$ admits a CP refutation of length at most $n^{O(1)}$.
\item If $F$ is \emph{unsatisfiable}, then $\calA(F)$ requires CP refutations of length at least $2^{n^{\Omega(1)}}$.
\end{itemize}
\end{theorem}
Consequently, it is $\NP$-hard to approximate the minimum CP proof length up to a factor of $2^{n^{\varepsilon}}$ for some $\varepsilon > 0$. In particular, CP is not automatable unless $\P=\NP$.

\subsection{Tree-like result} \label{sec:tree-result}

Our second result is a similar theorem for \emph{tree-like} Cutting Planes. However, we need a stronger hardness assumption (which is morally necessary; see \autoref{sec:tree-overview}) that we now formulate.

An \emph{$n$-set system} is a collection $\calS = \{S_1,\ldots,S_n\}$ where $S_i\subseteq[n]$ for each $i\in[n]$. A subset $H \subseteq [n]$ is a \emph{hitting set} for $\calS$ if $H \cap S_i \neq \emptyset$ for all~$i \in [n]$. The \emph{hitting set number} of $\calS$, denoted~$\gamma(\calS)$, is the minimum size of a hitting set for~$\calS$. The $k$-$\GHS$ promise problem is to distinguish between the cases $\gamma(\calS) \leq k$ versus $\gamma(\calS) \geq k^2$. A trivial algorithm can solve this problem in time $n^{O(k)}$. It is conjectured that there are no nontrivial algorithms for $k$ as large as $(1-\epsilon)\log n$. Under the Exponential-Time Hypothesis~\cite{eth}, the problem is known to be hard up to $k \leq (\log\log n)^{1-o(1)}$~\cite{Lin2019}. We need an assumption that is stronger by a hair's breadth.

\begin{conjecture} \label{conj:hitting-set}
The~$k$-$\GHS$ problem requires time $n^{\Omega(k)}$ for some $k=k(n)$ with
\begin{equation} \label{eq:k}
\omega(\log\log n)~\leq~k(n)~\leq~\log^{1/3}n. \tag{$\dagger$}
\end{equation}
\end{conjecture}

Our second main result says that tree-like CP is not automatable under \autoref{conj:hitting-set}.
\begin{theorem}[Tree-like] \label{thm:main-tree}
Let $k=k(n)$ satisfy \eqref{eq:k}. There is an $n^{o(k)}$-time algorithm ${\cal A}$ that on input an $n$-set system $\calS$,
outputs a CNF formula $\calA(\calS)$ such that:
\begin{itemize}
\item If $\gamma(\calS) \leq k$, then ${\cal A}(\calS)$ admits a tree-like CP refutation of length at most 
$n^{o(k)}$.
\item If $\gamma(\calS) \geq k^2$, then ${\cal A}(\calS)$ requires tree-like CP refutations
of length at least $n^{\omega(k)}$.
\end{itemize}
\end{theorem}

\section{Overview of proofs} \label{sec:overview}

In this section, we explain why both of our main results (dag-like and tree-like) follow from appropriate kinds of \emph{lifting theorems}. Abstractly speaking, a \emph{lifting theorem} is a tool that translates a lower-bound result for a weak model of computation (for us, Resolution) into an analogous lower-bound result for a strong model of computation (for us, Cutting Planes). Starting with Raz and McKenzie~\cite{raz99separation} such theorems now exist for an enormous variety of computational models. In proof complexity alone, prior examples of lifting applications include~\cite{bonet00relative,huynh12virtue,Goos2018,rezende16limited,garg18monotone,goos19monotone,Rezende2019}. We provide two more.

\subsection{Dag-like case}
Our proof of \autoref{thm:main-dag} builds directly on top of the breakthrough of Atserias and M\"uller~\cite{Atserias2019}. Given an $n$-variate 3-CNF formula $F$, they construct a formula $\myRef(F)$, which is an intricate CNF encoding of the claim ``$F$ admits a short Resolution refutation.'' Luckily, the exact details of~$\myRef(F)$ are not important for us. We only need a few high-level properties of their construction.

\paragraph{Block-width.}
The variables of $\myRef(F)$ come partitioned into some number of \emph{blocks}. Given a clause $D$ over the variables of $\myRef(F)$, we define its \emph{block-width} as the number of blocks that $D$ \emph{touches}, that is, contains a variable (or its negation) from that block. The \emph{block-width} of a Resolution refutation is the maximum block-width of any of its clauses.

\begin{lemma}[Atserias--M\"uller~\cite{Atserias2019}] \label{lem:am}
There is a polynomial-time algorithm that on input an $n$-variate 3-CNF formula $F$ outputs an unsatisfiable\footnote{Strictly speaking, $\myRef(F)$, as defined in \cite{Atserias2019}, may sometimes be satisfiable, in which case its Resolution width/length complexity is understood as $\infty$. However this case is equivalent to our reformulation, as we can guarantee that $\myRef(F)$ is always unsatisfiable by consider instead the CNF formula $\myRef(F)\land T$ where $T$ is some formula over disjoint variables known to require large width (e.g., Tseitin contradictions~\cite{urquhart87hard}).} CNF formula $\myRef(F)$ such that
\begin{itemize}
	\item If $F$ is \emph{satisfiable}, then $\myRef(F)$ admits a $n^{O(1)}$-length $O(1)$-block-width Resolution refutation.
	\item If $F$ is \emph{unsatisfiable}, then $\myRef(F)$ requires Resolution refutations of block-width at least $n^{\Omega(1)}$.
\end{itemize}
\end{lemma}

Atserias and M\"uller finish their proof by modifying $\myRef(F)$ slightly via \emph{relativization}, an operation due to Danchev and Riis~\cite{Dantchev2003} (see also~\cite{Garlik2019}). What this operation achieves is to turn a formula requiring block-width $b$ into a formula requiring Resolution length $2^{\Omega(b)}$. If $F$ is unsatisfiable, relativized-$\myRef(F)$ will have exponential length complexity. On the other hand, if $F$ is satisfiable, relativized-$\myRef(F)$ continues to have a short Resolution refutation, inherited from $\myRef(F)$.

In this paper, in order to make $\myRef(F)$ hard for Cutting Planes (when $F$ is unsatisfiable), we will modify the formula by block-wise \emph{composing (aka lifting)} it with a small gadget, an operation similar to relativization.

\paragraph{Lifting width.}
Recently, Garg et al.~\cite{garg18monotone} introduced a new lifting-based lower-bound technique for Cutting Planes: they showed how to lift Resolution width to Cutting Planes length. Namely, if $F$ is an $n$-variate formula requiring Resolution width $w$, then for a careful choice of a \emph{gadget} $g\colon\{0,1\}^m\to\{0,1\}$, $m=n^{O(1)}$, the composed formula $F\circ g^n$---obtained from $F$ by substituting each of its variables with a copy of $g$---has Cutting Planes length complexity $n^{\Theta(w)}$.

What would happen if we tried to apply the lifting result of~\cite{garg18monotone} to the formula $\myRef(F)$? When $F$ is unsatisfiable, we indeed do get (using width $\geq$ block-width) that $\myRef(F)\circ g^n$ requires exponential-length CP refutations. However, when $F$ is satisfiable, even though $\myRef(F)$ is promised to have block-width $O(1)$, its usual width still turns out to be $n^{\Omega(1)}$. Therefore the composition with $g$ would blow up the length complexity, not creating the desired gap in CP proof length.

\paragraph{Lifting block-width.}

Our idea, in short, is to build on~\cite{garg18monotone} and prove a lifting theorem for block-width (instead of width). Suppose $F$ is a formula whose $n\ell$ variables are partitioned into $n$ many blocks of $\ell$ variables each (typically $\ell=n^{\Theta(1)}$). We will consider compositions $F\circ g_\ell^n$ with a \emph{multi-output} gadget $g_\ell\colon\{0,1\}^m\to\{0,1\}^\ell$, one gadget for each block; see \autoref{sec:blocks} for the formal definition. Below, $\res(\,\cdot\,)$ denotes Resolution length complexity, $\cut(\,\cdot\,)$ denotes Cutting Planes length complexity, and $\bw(\,\cdot\,)$ denotes Resolution block-width complexity.
\begin{theorem}[Block lifting] \label{thm:block-lift-cp}
Fix an unsatisfiable CNF formula $F$ having $n$ many blocks of $\ell$ variables each. There is a gadget $g_\ell\colon\{0,1\}^m\to\{0,1\}^\ell$ where $m\coloneqq (n\ell)^{\Theta(1)}$ such that
\[
m^{\Omega(\bw(F))}
~\leq~
\cut(F\circ g_\ell^n)
~\leq~
\res(F\circ g_\ell^n)
~\leq~
m^{O(\bw(\Pi))}\cdot |\Pi|,
\]
where $\Pi$ is any Resolution refutation of $F$ of length $|\Pi|$ and block-width $\bw(\Pi)$.
\end{theorem}

Our main dag-like theorem (\autoref{thm:main-dag}) now follows immediately by combining \autoref{lem:am} and \autoref{thm:block-lift-cp}. Namely, consider the algorithm $\calA$ that on input an $n$-variate $3$-CNF formula $F$ outputs the CNF formula $\calA(F)\coloneqq \myRef(F)\circ g_\ell^k$ where $\myRef(F)$ has $k\leq n^{O(1)}$ many blocks with $\ell\leq n^{O(1)}$ variables each. We only need to note that this composed formula is constructible in polynomial time, which will be evident from the formal definition; see \autoref{fact:ref} in \autoref{sec:cnf-enc}. Therefore, to prove \autoref{thm:main-dag} it remains to prove \autoref{thm:block-lift-cp}, which we do in \autoref{sec:dag-lifting}.

\paragraph{Relation to monotone circuits.}

To conclude this subsection, we offer some philosophical musings on the techniques used to prove \autoref{thm:block-lift-cp}. Non-automatability results for Cutting Planes have been elusive in part because of the limitations of existing techniques to prove lower bounds on refutation length (as required by the second item in \autoref{thm:main-dag}). The only technique available for some twenty years has been \emph{monotone feasible interpolation}~\cite{bonet97lower,krajicek97interpolation,hrubes18note}, which translates lower bounds for (real) monotone circuits to lower bounds on Cutting Planes length. Historically, the downside with the technique was that it only seemed to apply to highly specialized formulas (e.g., clique-vs-coloring). However, the technique was recently extended to handle a more general class of formulas, random $\Theta(\log n)$-CNFs~\cite{hrubes17randomformulas,fleming17randomcnf}. The only other available lower-bound technique is the aforementioned lifting theorem~\cite{garg18monotone}. That technique is also powerful enough to prove lower bounds not only on CP length, but also on monotone circuit size. (Whether lifting should be classified under monotone interpolation is up for debate, since this depends on how broadly one defines monotone interpolation.)

In contrast, our \autoref{thm:block-lift-cp} is \emph{not} proved through monotone circuit lower bounds, but through a new weaker model of computation, dubbed \emph{simplex-dags} in \autoref{sec:sim-dag}. At the heart of monotone interpolation is a characterization of monotone circuits by a two-party communication game~\cite{Razborov1995,pudlak10extracting,sokolov17dag}. In this language, our \autoref{thm:block-lift-cp} is obtained not by studying a two-party communication model, but rather a \emph{multi-party} model. Considering a large number of communicating parties is what allows us to analyze multi-output gadgets; we do not know how to do this with only two parties.

\subsection{Tree-like case} \label{sec:tree-overview}

Our proof of \autoref{thm:main-tree} builds on the important paper by Alekhnovich and Razborov~\cite{Alekhnovich2008} (which has been followed up by \cite{Galesi2010,Mertz2019}). They show that tree-like Resolution is not automatable assuming the fixed parameter hierarchy does not collapse (which is implied by the Exponential-Time Hypothesis). Since tree-like Resolution proofs {\it can} be found in quasipolynomial-time (we say tree-like Resolution is \emph{quasipolynomially} automatable), they need to assume more than $\NP$-hardness. Our results will inherit this need for a stronger assumption (namely, \autoref{conj:hitting-set}), even though tree-like CP is not known to be quasipolynomially automatable.

The reduction of Alekhnovich and Razborov is somewhat complicated, but luckily we will only need as our starting point the following lemma from the follow-up work~\cite{Mertz2019}.

\begin{lemma}[Mertz et al.~\cite{Mertz2019}]\label{lem:mpw}
Let $k \leq \log^{1/3} n$. There is a polynomial-time algorithm $\calB$ that on input a $n$-set system $\calS$ outputs an unsatisfiable $O(\log n)$-CNF formula $\calB(\calS)$ such that
\begin{itemize}
\item If $\gamma(\calS) \leq k$, then $\calB(\calS)$ admits a Resolution refutation of depth $O(\log n)$.
\item If $\gamma(\calS) \geq k^2$, then $\calB(\calS)$ requires Resolution refutations of depth $\Omega(k \log n)$.
\end{itemize}
\end{lemma}

To prove \autoref{thm:main-tree}, our plan is once again to compose the formula $\calB(\calS)$ with a (single-output-bit) gadget in order to lift the Resolution depth gap in \autoref{lem:mpw} into tree-like CP length gap. To this end, we develop a new lifting theorem for ``small'' gadgets.

\paragraph{Limitations of existing methods.}
Let $F$ be an unsatisfiable $n$-variate formula with Resolution depth complexity $\d(F)$. Existing lifting theorems~\cite{bonet00relative,rezende16limited} when applied to $F$ would require a gadget $g\colon\{0,1\}^{\poly(n)}\to\{0,1\}$ that can be computed by a decision tree of depth $\Theta(\log n)$ and hence of size $\poly(n)$. Writing $\resTree(\,\cdot\,)$ for tree-like Resolution length complexity, and $\cutTree(\,\cdot\,)$ for tree-like CP length complexity, the lifting theorems~\cite{bonet00relative,rezende16limited} show
\begin{equation} \label{eq:prior}
\cutTree(F\circ g^n)
~=~
\resTree(F\circ g^n)^{\Theta(1)}
~=~
n^{\Theta(\d(F))}.
\end{equation}
The base of the exponent above (namely, $\poly(n)$) is the decision tree size of~$g$. If we applied \eqref{eq:prior} to \autoref{lem:mpw}, we would only end up with a length gap of $n^{O(\log n)}$ versus $n^{\omega(\log n)}$. But these lengths---and hence running times for the automating algorithm---are enough to solve the $k$-$\GHS$ problem, which prevents us from getting a hardness result.

\paragraph{Small gadget lifting.}
What we need is a lifting theorem for small gadgets, that is, gadgets computed by small decision trees. It is an important open problem whether tree-like lifting is possible with a \emph{constant-size} gadget. In this paper, we are able to use a gadget of decision-tree size depending only on the quantity we want to lift, namely $\d(F)$, and not depending on the number of variables $n$ of $F$. Our lifting theorem can be seen as a generalization of previous ones, which handled the case $\d(F)=n^{\Omega(1)}$, and can also be viewed as a step towards proving a lifting theorem for significantly smaller gadgets (eventually, constant-size).
\begin{theorem}[Small gadget lifting] \label{thm:small-lift}
For every $m$ there exists a gadget $g\colon\{0,1\}^{\poly(m)}\to\{0,1\}$ of query complexity $O(\log m)$ such that for every unsatisfiable $n$-variate CNF formula $F$,
\[
m^{\Theta(\min(\d(F), m))}
~\leq~
\cutTree(F\circ g^n)
~\leq~
\resTree(F\circ g^n)
~\leq~
m^{O(\d(F))}.
\]
\end{theorem}

Our main tree-like theorem (\autoref{thm:main-tree}) now follows by combining \autoref{lem:mpw} and \autoref{thm:small-lift}. Indeed, choose $m\coloneqq \log^2 n$ and consider the algorithm $\calA$ that on input an $n$-set system $\calS$ outputs the formula $\calA(\calS)\coloneqq \calB(\calS)\circ g^{n'}$ where $\calB(\calS)$ has $n'=n^{O(1)}$ variables. We have
\begin{align*}
\begin{tabular}{lll}
$\gamma(\calS)~\leq~k$
&$\implies$
&$\cutTree(\calA(\calS))~\leq~m^{O(\log n)}~=~n^{O(\log\log n)}~\leq~n^{o(k)},$\\
$\gamma(\calS)~\geq~k^2$
&$\implies$
&$\cutTree(\calA(\calS))~\geq~m^{\Omega(k \log n)}~=~n^{\Omega(k\log\log n)} ~\geq~n^{\omega(k)}.$
\end{tabular}
\end{align*}
Finally, we note that the composed formula $\calB(\calS)\circ g^{n'}$ can be constructed in time $n^{o(k)}$. This will be evident from the formal definition (see \autoref{fact:ind}), but the intuition is as follows. Each $O(\log n)$-width clause of $\calB(\calS)$ will turn into a whole family $O(\log m\log n)$-width clauses for $\calB(\calS)\circ g^{n'}$. The family for a particular clause $D$ is obtained by replacing each literal of $D$ in all possible ways by an $O(\log m)$-length root-to-leaf path (of which there are $2^{O(\log m)}$ many) in the decision tree for $g$. Altogether this will yield $|\calB(\calS)|\cdot (2^{O(\log m)})^{O(\log n)} = n^{O(1)}\cdot n^{o(k)}=n^{o(k)}$ many clauses. Therefore, to prove \autoref{thm:main-tree} it remains to prove \autoref{thm:small-lift}, which we do in \autoref{sec:tree-lifting}.

\paragraph{Relation to real protocols.}
The lower bound in \autoref{thm:small-lift} holds not only for tree-like Cutting Planes but also for a stronger model of computation, \emph{real communication protocols}~\cite{Kraj98}. This is not surprising: all existing lower bounds on tree-like CP length have been proved through real protocols (or the even more powerful model of randomized protocols). In a nutshell, our proof of \autoref{thm:small-lift} extends the techniques in a long line of work on tree-like lifting~\cite{raz99separation,bonet00relative,goos15deterministic,rezende16limited,goos17bpp,CFKMP}, optimizing the argument in order to get rid of the dependence on the input size~$n$. A detailed overview is given in \autoref{sec:tree-lifting}.

\section{Dag-like definitions} \label{sec:dag-definitions}

In this paper, we adopt the standard \emph{top-down} view of proofs~\cite{pudlak00proofs,atserias08combinatorial}. Namely, we interpret a refutation of an $n$-variate CNF formula $F\coloneqq \land_{i\in[m]} D_i$ as a way of solving the associated \emph{falsified-clause search problem} $S_F\subseteq\{0,1\}^n\times [m]$. The problem $S_F$ is, on input a truth assignment $x\in\{0,1\}^n$, to find a clause $D_j$, $j\in[m]$, falsified by $x$, that is, $D_j(x)=0$. For example, tree-like Resolution refutations of $F$ are equivalent to decision trees solving $S_F$~\cite{lovasz95search}. We proceed to formalize this for dag-like models. The material in \autoref{sec:dag-models} is standard. \autoref{sec:sim-dag} introduces a novel model, \emph{simplex-dags}, for which we develop a lifting theorem in \autoref{sec:dag-lifting}.

\subsection{Standard models} \label{sec:dag-models}

\paragraph{Abstract dags.}
Fix an abstract search problem $S\subseteq\calI\times\calO$, that is, on input $x\in\calI$ the goal is to find some $o\in S(x)\coloneqq\{o\in \calO: (x,o)\in S\}$. We always work with \emph{total} search problems where $S(x)\neq\emptyset$ for all $x\in\calI$. Fix also a family $\calF$ of functions $\calI\to\{0,1\}$. An \emph{$\calF$-dag} solving $S$ is a directed acyclic graph of out-degree $\leq 2$ where each vertex $v$ is associated with a function $f_v\in \calF$ (here $f^{-1}(1)$ is sometimes called the \emph{feasible set} for $v$) satisfying the following.
\begin{itemize}
\item \emph{Root.} There is a designated root vertex $v$ (in-degree $0$) that satisfies $f_v\equiv 1$.
\item \emph{Non-leaf.} Every non-leaf $v$ with children $u,u'$ (perhaps $u=u'$) has $f_v^{-1}(1)\subseteq f_u^{-1}(1)\cup f_{u'}^{-1}(1)$.
\item \emph{Leaves.} For every leaf $v$ there is some output $o\in\calO$ such that $f_v^{-1}(1)\subseteq S^{-1}(o)$.
\end{itemize}
The \emph{size} of an $\calF$-dag is its number of vertices.

\paragraph{Decision-dags and Resolution.}
Consider instantiating the above template with the $n$-bit input domain $\calI\coloneqq\{0,1\}^n$ and taking $\calF$ to be the set of all conjunctions over the literals $x_1,\bar{x}_1,\ldots,x_n,\bar{x}_n$. We call such $\calF$-dags simply \emph{decision-dags}. Apart from the size of a decision-dag another important measure is its \emph{width}: the maximum width of a conjunction used. We define
\begin{align*}
\decDag(S)~&\coloneqq~\text{least \emph{size} of a decision-dag solving $S$},\\
\w(S)~&\coloneqq~\text{least \emph{width} of a decision-dag solving $S$}.
\end{align*}
When specialized to unsatisfiable CNF search problems $S=S_F$, we recover the usual Resolution proof system. Indeed, $\decDag(S_F)$ equals $\res(F)$, the length required to refute $F$ in Resolution, and $\w(S_F)$ equals the Resolution width complexity of $F$ (famously studied in~\cite{ben-sasson01short}).

\paragraph{LTF-dags and Cutting Planes.}
Consider instantiating $\calI\coloneqq\{0,1\}^n$ and taking $\calF$ to be the set of all $n$-bit \emph{linear threshold functions} (LTFs). Recall that an $f\in\calF$ is defined by a vector $a\in \R^{n+1}$ such that $f(x)=1$ iff $\sum_{i\in[n]} a_ix_i \geq a_{n+1}$. We call such $\calF$-dags simply \emph{LTF-dags}, and define
\[
\ltfDag(S)~\coloneqq~\text{least \emph{size} of an LTF-dag solving $S$}.
\]
When specialized to $S=S_F$, we recover the semantic Cutting Planes proof system. Indeed, $\ltfDag(S_F)$ equals $\cut(F)$, the length required to refute $F$ in semantic Cutting Planes.

\subsection{Simplex-dags} \label{sec:sim-dag}
We now introduce a new type of dag, for which our dag-like lifting theorem is formulated (\autoref{sec:dag-lifting}). Let $k\geq 1$ and consider a fixed $k$-partite input domain $\calI\coloneqq \calI_1 \times\cdots\times \calI_k$. We say that a function $f\colon \calI_1 \times\cdots\times \calI_k \to\{0,1\}$ is \emph{monotone} (up to an ordering of the parts $\calI_i$; aka \emph{unate}) iff each set $\calI_i$ admits a total order $\preceq_i$ such that $f(x)\leq f(y)$ for every pair $x\preceq y$ (meaning $x_i\preceq_i y_i$ for all $i\in[k]$). For example, every $n$-bit LTF is monotone as an $n$-partite function: the orderings are determined by the signs of the coefficients appearing in the linear form defining $f$. We also say that a subset $A\subseteq \calI_1 \times\cdots\times\calI_k$ is a \emph{(combinatorial) $k$-simplex} if its indicator function is monotone. Let $\calF$ be the set of monotone functions over $\calI_1 \times\cdots\times\calI_k$; we emphasize that any two $f,f'\in\calF$ may not agree on the ordering of any part $\calI_i$. We call such $\calF$-dags simply \emph{simplex-dags}, and define
\[
\simDag(S)~\coloneqq~\text{least \emph{size} of a simplex-dag solving $S$}.
\]

\paragraph{Relation to other models.}
Simplex-dags are a natural $k$-party generalization of the bipartite case $k=2$, which was called \emph{triangle-dags} in \cite{garg18monotone}. Triangle-dags in turn are equivalent to real circuits and real dag-like protocols~\cite{haken99exponential,pudlak97lower,hrubes18note}. Our motivation to consider multi-party models is that they can be vastly weaker than two-party models. Hence one expects it to be easier to prove lower bounds for $k$-simplex-dags when $k$ is large. For a toy example, consider the $n$-bit $\Xor_n$ function. It is easy to compute for traditional two-party communication protocols regardless of how the $n$ bits are split between the two players. By contrast, for $n$ parties, each holding one input bit, $\Xor_n$ is hard to compute.

\subsection{Relationships}
The complexity measures introduced so far are related as follows:
\[
\simDag(S_k)~\leq~\ltfDag(S_n)~\leq~\decDag(S_n)~\leq~n^{O(\w(S_n))}.
\]
Here $S_n\subseteq\{0,1\}^n\times\calO$ is any $n$-bit search problem, and $S_k\subseteq \{0,1\}^{I_1} \times\cdots\times\{0,1\}^{I_k}\times\calO$ is a $k$-partite version of $S_n$ obtained from an arbitrary partition $I_1\sqcup\cdots\sqcup I_k=[n]$. The first inequality follows by noting that each LTF $f$, defined by $\sum_i a_ix_i \geq a_{n+1}$, is a monotone $k$-partite function when the $i$-th part $\{0,1\}^{I_i}$ is ordered according to the partial sum $\sum_{i\in I_i} a_ix_i$ (breaking ties arbitrarily). The second inequality follows since every conjunction is an LTF. The last inequality is standard: the length of any width-$w$ Resolution refutation can be made $n^{O(w)}$ by eliminating repeated clauses (and the same construction works for arbitrary search problems).

\subsection{Blocks} \label{sec:blocks}

\paragraph{Block width.}
Let $S\subseteq (\{0,1\}^\ell)^n \times\calO$ be any search problem whose $n\ell$ input bits are partitioned into $n$ blocks of $\ell$ bits each. For every conjunction $C$ over the variables of $S$, we define the \emph{block-width} of $C$ as the maximum number of blocks that $C$ \emph{touches}, that is, contains a variable (or its negation) from a block. We define the \emph{block-width} of a decision-dag solving $S$ as the maximum block-width over all conjunctions in the dag. Finally, we define
\[
\bw(S)~\coloneqq~\text{least \emph{block-width} of a decision-dag solving $S$}.
\]

\paragraph{Block composition.}
The \emph{column-index} gadget $\Ind_{\ell\times m}\colon[m]\times \{0,1\}^{\ell\times m}\to\{0,1\}^\ell$ is defined by $\Ind_{\ell\times m}(x,y) \coloneqq \text{``$x$-th column of $y$''}$. We call $y\in\{0,1\}^{\ell\times m}$ the \emph{matrix} and $x\in[m]$ the \emph{pointer} (for decision-dags, we tacitly encode the elements of $[m]$ in binary as $\log m$-bit strings.). Letting $S\subseteq(\{0,1\}^\ell)^n\times\calO$ be as above, we define a composed search problem
\begin{equation} \label{eq:s-ind}
S\circ \Ind_{\ell\times m}^n ~\subseteq~ [m]^n\times(\{0,1\}^{\ell\times m})^n\times\calO.
\end{equation}
Namely, on input $(x,y)\in [m]^n\times(\{0,1\}^{\ell\times m})^n$ the goal is to find an output $o\in S(z)$ for $z \coloneqq (\Ind_{\ell\times m}(x_1,y_1),\ldots,\Ind_{\ell\times m}(x_n,y_n))\in(\{0,1\}^\ell)^n$. We shall view the composition~\eqref{eq:s-ind} as an $(1+n\ell)$-partite search problem by repartitioning the input domain as
\[\textstyle
[m]^n\times(\{0,1\}^{\ell\times m})^n~=~
\calX\times\prod_{(i,j)\in[n]\times[\ell]} \calY^{ij}
\enspace\qquad\text{where}\quad
\setlength\tabcolsep{.3em}
\left\{
\begin{tabular}{ll}
$\calX$&$\coloneqq~[m]^n$\\
$\calY^{ij}$&$\coloneqq~\{0,1\}^m$.
\end{tabular}\right.
\]
Here we think of player $\Alice$ as holding $x\in \calX$, and for $(i,j)\in[n]\times[\ell]$, player $\Bob^{ij}$ as holding $(y_i)_j\in \calY^{ij}$, that is, the $j$-th row of the $i$-th matrix $y_i$.

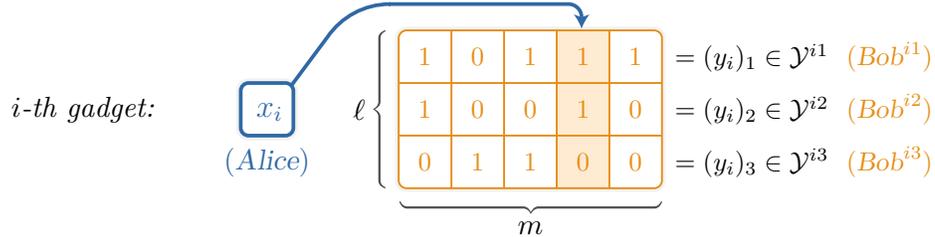
\begin{center}
\begin{lpic}[]{figs/index(.35)}
\large
\lbl[c]{31,49;\color{myBlue}$x_i$}
\normalsize
\lbl[c]{30,30;\color{myBlue} ($\Alice$)}
\lbl[c]{-40,50;\itshape $i$-th gadget:}
\lbl[c]{65,50;$\ell$}
\lbl[c]{130,5;$m$}
\small
\lbl[l]{185,70;$=(y_i)_1 \in \calY^{i1}$ \color{myGold} ~($\Bob^{i1}$)}
\lbl[l]{185,50;$=(y_i)_2 \in \calY^{i2}$ \color{myGold} ~($\Bob^{i2}$)}
\lbl[l]{185,30;$=(y_i)_3 \in \calY^{i3}$ \color{myGold} ~($\Bob^{i3}$)}
\small
\lbl[c]{90,30;\color{myGold} 0}
\lbl[c]{90,50;\color{myGold} 1}
\lbl[c]{90,70;\color{myGold} 1}
\lbl[c]{110,30;\color{myGold} 1}
\lbl[c]{110,50;\color{myGold} 0}
\lbl[c]{110,70;\color{myGold} 0}
\lbl[c]{130,30;\color{myGold} 1}
\lbl[c]{130,50;\color{myGold} 0}
\lbl[c]{130,70;\color{myGold} 1}
\lbl[c]{150,30;\color{myGold} 0}
\lbl[c]{150,50;\color{myGold} 1}
\lbl[c]{150,70;\color{myGold} 1}
\lbl[c]{170,30;\color{myGold} 0}
\lbl[c]{170,50;\color{myGold} 0}
\lbl[c]{170,70;\color{myGold} 1}
\end{lpic}
\end{center}

\subsection{CNF encoding} \label{sec:cnf-enc}

We just defined block-composed search problems $S\circ \Ind^n_{\ell\times m}$, but how can we translate such objects back to CNF formulas? The standard recipe is as follows. Fix any search problem $S\subseteq \{0,1\}^n\times\calO$ (not necessarily of a composed form). A \emph{certificate} for $(x,o)\in S$ is a partial assignment $\rho\in\{0,1,*\}^n$ consistent with $x$ such that for any $y$ consistent with $\rho$ we have $(y,o)\in S$. The size of $\rho$ is the number of its fixed (non-$*$) coordinates. The \emph{certificate complexity} of~$S$ is the maximum over all inputs $x\in\{0,1\}^n$ of the minimum over all $o\in S(x)$ of the least size of a certificate for $(x,o)$. For example, if $F$ is an unsatisfiable $k$-CNF formula, then $S_F$ has certificate complexity at most $k$. Conversely, any total search problem $S$ of certificate complexity $k$ contains the search problem $S_F$ associated with some unsatisfiable $k$-CNF formula $F$ as a subproblem ($S$ is at least as hard as $S_F$). Namely, consider $F\coloneqq\bigwedge_x \neg C_x$ where $C_x$ is the conjunction that checks if the input is consistent with some fixed size-$k$ certificate for $x$. Note that $F$ is unsatisfiable because $S$ is total.

For unbounded-width CNF formulas (such as Atserias--M\"uller's $\myRef(F)$), we need to interpret the above recipe with care. Indeed, fix any unsatisfiable (unbounded-width) CNF formula $F$ with~$|F|$ many clauses and such that its $n\ell$ variables are partitioned into $n$ blocks of $\ell$ variables each. Denote by $b$ the maximum block-width of a clause of $F$. Then every clause $D$ of $F$ gives rise to a family of certificates for $S_F\circ \Ind^n_{\ell\times m}$. Namely, a certificate in the family for $D$ consists of at most $b\log m$ bits (reading $b$ many pointer values associated with the blocks of $D$) together with $|D|$ many bits read from the pointed-to columns. Thus, altogether, we get at most $|F| m^b$ many certificates, at least one for each input to $S_F\circ \Ind^n_{\ell\times m}$. We define $F\circ \Ind^n_{\ell\times m}$ as the formula obtained by listing all these certificates (more precisely, the disjunctions that are the negations of the certificates).

The formula $\myRef(F)$ of Atserias and M\"uller is such that its clauses have block-width~$3$~\cite[Appendix~A]{Atserias2019}. Hence $\myRef(F)\circ \Ind^n_{\ell\times m}$ has size $n^{O(1)}$ and moreover it is polynomial-time constructible.%
\begin{fact} \label{fact:ref}
Given an $n$-variate $3$-CNF $F$, we can construct $\myRef(F)\circ \Ind^n_{\ell\times m}$ in polynomial time.
\end{fact}

\section{Dag-like lifting} \label{sec:dag-lifting}

The purpose of this section is to prove our block-lifting theorem (\autoref{thm:block-lift-cp}), which would complete the proof of our main dag-like result (\autoref{thm:main-dag}). We restate the block-lifting theorem using the search-problem-centric language of \autoref{sec:dag-definitions}. Then \autoref{thm:block-lift-cp} is the special case $S\coloneqq S_F$.
\begin{theorem}[Block lifting] \label{thm:block-lift}
Let $S\subseteq(\{0,1\}^\ell)^n\times\calO$ be any search problem. For $m\coloneqq (n\ell)^5$ we have
\[
m^{\Omega(\bw(S))}
~\leq~
\simDag(S\circ\Ind_{\ell\times m}^n)
~\leq~
\decDag(S\circ\Ind_{\ell\times m}^n)
~\leq~
m^{O(\bw(\Pi))}\cdot|\Pi|,
\]
where $\Pi$ is any decision-dag solving $S$ of size $|\Pi|$ and block-width $\bw(\Pi)$.
\end{theorem}

\paragraph{Upper bound.}
The last inequality is the trivial part of \autoref{thm:block-lift}. We only sketch it here. Given a decision-dag $\Pi$ for $S$, we construct a decision-dag $\Pi'$ for $S\circ\Ind_{\ell\times m}^n$. For every block-width-$b$ conjunction $C$ in $\Pi$, there corresponds a family of exactly $m^b$ many conjunctions in~$\Pi'$. Namely, the family is constructed by replacing each positive literal $x_{ij}$ (resp.\ negative literal $\bar{x}_{ij}$) of $C$ with a sequence of $\log m+1$ many literals that witness the $j$-th output bit of the $i$-th gadget being $1$ (resp.~$0$). If $C$ has children $C'$, $C''$ that only touch blocks touched by $C$, then every conjunction in the family for $C$ can be directly connected to the families of $C'$, $C''$. However, if $C'$, $C''$ touch some block $i$ (there can be at most one) that is untouched by $C$, then the family for $C$ is connected to the families of $C'$, $C''$ via decision trees that query the pointer value of the $i$-th gadget. We have $|\Pi'|\leq m^{O(\bw(\Pi))}\cdot|\Pi|$, as desired.

\paragraph{Lower bound.}
The first inequality is the nontrivial part of \autoref{thm:block-lift}. Our proof follows closely the plan from~\cite{garg18monotone}. However, the proof here is in many ways simpler than the original one. The reason is that we work with multi-party objects (high-dimensional boxes and simplices) rather than two-party objects (rectangles and triangles). For example, one of the key technical lemmas, \autoref{lem:rho-like} (``$\rho$-structured boxes are $\rho$-like''), admits a short proof in our multi-party setting, whereas the original lemma for two parties required a long proof involving Fourier analysis. The rest of this section is concerned with proving the simplex-dag lower bound.

\subsection{Subcubes from simplices}

Let $\rho\in(\{0,1\}^\ell\cup\{*\})^n$ be a partial assignment that assigns each of the $n$ blocks either an $\ell$-bit string or the star symbol. We denote by $\free(\rho)\subseteq[n]$ the subset of blocks assigned a star, and define $\fix(\rho)\coloneqq [n]\smallsetminus\free(\rho)$. The subcube of strings consistent with $\rho$ is $\Cube(\rho) \coloneqq \{ z\in (\{0,1\}^\ell)^n : z_i=\rho_i,\forall i\in\fix(\rho)\}$. For any set $R\subseteq \calX\times\prod \calY^{ij}$ we say that
\[
\textit{``$R$ is $\rho$-like''}
\qquad \text{iff} \qquad
\Ind_{\ell\times m}^n(R)~=~\Cube(\rho).
\]

We formulate a sufficient condition for $R$ to be $\rho$-like in case $R$ is a \emph{box}, that is a product set.
\begin{definition}[Random variables] \label{def:entropy}
For a random variable $\x\in \calX$ we define its \emph{min-entropy} by $\Hmin(\x)\coloneqq\min_x \log(1/\Pr[\,\x=x\,])$. When $\x$ is chosen from a set $\calX^k$ that is partitioned into $k$ blocks, we define its \emph{blockwise min-entropy} by $\min_{\emptyset\neq S \subseteq [k]} \frac{1}{|S|}\Hmin(\x_S)$ where $\x_S$ is the marginal distribution of $\x$ over blocks $S$. We also define the \emph{deficiency} of $\x\in\calX$ by $\Dmin(\x)\coloneqq \log|\calX| - \Hmin(\x)\geq 0$. For convenience, if $X$ is a set, we denote by $\bm{X}\in X$ the random variable that is uniform over~$X$. In particular, for $X\subseteq\calX^n$ the notation $\bm{X}_I$ for $I\subseteq[n]$ means ``the marginal distribution over coordinates $I$ of the uniform distribution over $X$''. We use $X_I\coloneqq\{x_I:x\in X\}$ to mean the set that is the projection of $X$ onto coordinates $I$; thus $X_I$ is the support of $\bm{X}_I$.

\end{definition}

\begin{definition}[Structured boxes] \label{def:structured}
Let $R\coloneqq X\times\prod_{ij} Y^{ij}\subseteq \calX\times\prod_{ij}\calY^{ij}$ be a box and $\rho\in(\{0,1\}^\ell\cup\{*\})^n$ a partial assignment. We say $R$ is \emph{$\rho$-structured} if
\begin{enumerate}
\item Gadgets are fixed according to $\rho$: ~$\Ind_{\ell\times m}^{\fix(\rho)}(R_{\fix(\rho)})=\{\rho_{\fix(\rho)}\}$.
\item $X$ has entropy on the free blocks: ~$\X_{\free(\rho)}$ has blockwise min-entropy $\geq 0.9\cdot\log m$.
\item $Y^{ij}$ are large: ~$\Dmin(\Y^{ij})\leq m^{1/2}$ for $i\in\free(\rho)$, $j\in[\ell]$.
\end{enumerate}
\end{definition}

The following key lemma is the reason our dag-lifting result is formulated for $k$-simplex-dags for large $k$---we do not know how to prove a multi-output gadget lemma like this for $k=2$. (The paper \cite{garg18monotone} did it for $k=2$ and single-output gadgets.)

\begin{lemma} \label{lem:rho-like}
Let $R\coloneqq X\times\prod_{ij} Y^{ij}$ be $\rho$-structured. There is an $x\in X$ so that $\{x\}\times \prod_{ij} Y^{ij}$ is $\rho$-like.
\end{lemma}

\begin{proof}
Assume for simplicity that $\rho=*^n$. Thus our goal is to find an $x\in X$ such that $\Ind^n_{\ell\times m}(\{x\}\times \prod_{ij} Y^{ij})=(\{0,1\}^\ell)^n$. The key observation is that since each of the $n\ell$ output bits is determined by a different $\Bob^{ij}$, the output bits are independent: $\Ind^n_{\ell\times m}(\{x\}\times \prod_{ij} Y^{ij}) = \prod_{ij} \Ind_{1\times m}(\{x_i\}\times Y^{ij})$. Therefore it suffices to find an $x\in X$ such that for all $i\in[n]$, $j\in[\ell]$,
\begin{equation}\label{eq:one-output}
\textit{$x$ is ``good'' for $Y^{ij}$:} \qquad
\Ind_{1\times m}(\{x_i\}\times Y^{ij})~=~\{0,1\}.
\end{equation}
We claim that a uniform random choice $\x\in X$ satisfies all conditions \eqref{eq:one-output} with positive probability. Indeed, for a fixed $ij$, how many ``bad'' values $x_i\in[m]$ are there that fail to satisfy \eqref{eq:one-output}? Each bad value~$x_i$ implies that the $x_i$-th bit is fixed in $Y^{ij}$. But there can be at most $\Dmin(Y^{ij})\leq m^{1/2}$ fixed such bits. Using $\Hmin(\x_i)\geq 0.9\cdot\log m$ for $i\in[n]$ and recalling that $m = (n\ell)^5$ we have
\[
\Pr[\,\x_i \text{ is ``bad'' for } Y^{ij}\,]~\leq~ m^{1/2}\cdot 2^{-0.9\log m}~<~1/(n\ell).
\]
A union bound over all the $n\ell$ many conditions \eqref{eq:one-output} completes the proof.
\end{proof}

The following lemma is the culmination of this subsection: Every simplex can be partitioned into $\rho$-like pieces (and some error sets); see \autoref{fig:simlemma}. The lemma is a high-dimensional analogue of the \emph{Triangle Lemma} from~\cite{garg18monotone}. We defer the proof to \autoref{sec:proof-sim-lemma}.

\begin{figure}[t]
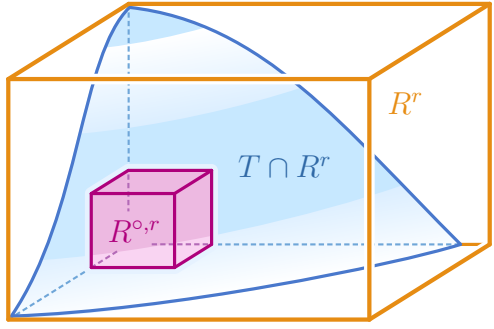

\begin{center}
\begin{lpic}[t(-6mm),b(-8mm)]{figs/sandwich(.4)}
\large
\lbl[c]{62,49.5;\color{myPurple}$R^{\circ,r}$}
\lbl[c]{112,71;\color{myBlue}$T\cap R^r$}
\lbl[c]{152,92;\color{myGold}$R^r$}
\end{lpic}
\end{center}
\caption{Structured case of \simlemma. The simplex $T$ is partitioned as $T=\bigsqcup_r T\cap R^r$ where each structured part is sandwiched between two $\rho^r$\!-structured boxes, $R^{\circ,r}\subseteq T\cap R^r\subseteq R^r$.}
\label{fig:simlemma}%
\end{figure}

\begin{sim-lemma}\label{lem:sim-lemma}
Let $T\subseteq \calX\times\prod_{ij}\calY^{ij}$ be a simplex and $k\geq 0$ an error parameter. There exists a disjoint box covering $\bigsqcup_r R^r\supseteq T$ and error sets $X^{\err}\subseteq\calX$, $Y^{\err,ij}\subseteq\calY^{ij}$, each of density $\leq 2^{-k}$, such that for each $r$ one of the following holds:
\begin{itemize}[label=$\bullet$,itemsep=1mm]
\item {\bf\slshape Structured case:}~ $R^r$ is $\rho^r$\!-structured for some $\rho^r$ that fixes $O(k/\log m)$ blocks. Moreover there exists an ``inner'' box $R^{\circ,r} \subseteq T\cap R^r$, which is also $\rho^r$\!-structured.
\item {\bf\slshape Error case:}~ $R^r$ is covered by error boxes: $R^r \subseteq X^{\err}\times\prod_{ij} \calY^{ij} \cup \bigcup_{ij} \calX\times Y^{\err,ij}\times\prod_{i'j'\neq ij}\calY^{i'j'}$.
\end{itemize}
Finally, a {\bf\itshape query alignment} property holds: for every $x \in \calX\smallsetminus X_{\err}$, there exists a subset $I_x \subseteq [n]$ with $|I_x| \le O(k/\log m)$ such that every ``structured'' $R^r$ intersecting $\{x\} \times \prod_{ij}\calY^{ij}$ has $\fix(\rho^r) \subseteq I_x$.
\end{sim-lemma}

\subsection{Simplified proof} \label{sec:simple-proof}

To prove (the first inequality of) \autoref{thm:block-lift}, fix a simplex-dag $\Pi$ solving $S\circ\Ind_{\ell\times m}^n$ of size $m^d$. Our goal is to construct a decision-dag $\Pi'$ solving $S$ that has block-width $O(d)$. We first present the proof under a simplifying assumption and then remove that assumption in \autoref{sec:errors}.
\begin{enumerate}[label=({\boldmath$*$})]
\item  \label{assumption}
\emph{Assumption:} If we apply \simlemma for $k\coloneqq 2d\log m$ to any simplex $T$ in $\Pi$, then each part in the produced partition $T=\bigsqcup_r T\cap R^r$ satisfies the ``structured case''.
\end{enumerate}

Using \ref{assumption}, apply \simlemma (for the above choice of $k$) to partition all simplicies $T$ in~$\Pi$. Each resulting structured part $T\cap R^r$ will correspond to a vertex in $\Pi'$ associated with the partial assignment (or conjunction)~$\rho^r$, that is, with feasible set $\Cube(\rho^r)$. Moreover, we will let the type (root/internal/leaf) of a vertex $T$ in $\Pi$ dictate the type of the resulting vertices $T\cap R^r$ in $\Pi'$. We will add more vertices to $\Pi'$ shortly in order to connect all the internal vertices, but so far $\Pi'$ already meets the \emph{root} and \emph{leaf} conditions of a decision-dag solving $S$, as we note next.

\paragraph{Step 1: Root and leaves.}
We may assume that for the root of $\Pi$, which is associated with the simplex $T\coloneqq \calX\times\prod_{ij}\times\calY^{ij}$, the \simlemma produces the trivial partition consisting of just one $*^n$-structured part, $T$ itself. Hence, the designated root of $\Pi'$ is defined as the sole part $T$ with an associated feasible set $\Cube(*^n)=(\{0,1\}^\ell)^n$. This meets the \emph{root} condition of a decision-dag.

Consider any part $R^{\circ,r}\subseteq T \cap R^r\subseteq R^r$ with an associated assignment $\rho^r$, arising from a \emph{leaf} $T$ of $\Pi$. Suppose $o\in\calO$ is a valid solution for $T$ in $\Pi$, that is, $T\subseteq (S\circ \Ind_{\ell\times m}^n)^{-1}(o)$, or equivalently, $\Ind_{\ell\times m}^n(T)\subseteq S^{-1}(o)$. We claim that $o$ is also a valid solution for the leaf $T \cap R^r$ in $\Pi'$:
\[
\Cube(\rho^r)
~=~ \Ind_{\ell\times m}^n(T\cap R^r)
~\subseteq~ \Ind_{\ell\times m}^n(T)
~\subseteq~ S^{-1}(o).
\]
Here the equality uses the fact that $T\cap R^r$ is $\rho^r$\!-like (it is sandwiched between two sets that are $\rho^r$\!-structured, and hence $\rho^r$\!-like by \autoref{lem:rho-like}). This meets the \emph{leaf} condition of a decision-dag.

\paragraph{Step 2: Internal.}
To complete the definition of $\Pi'$, consider a vertex associated with some part $R^{\circ}\subseteq T \cap R\subseteq R$, where $R^\circ$ and $R$ are $\rho$-structured, that arises from a \emph{non-leaf} simplex $T$ of $\Pi$. We connect this vertex to the vertices arising from $T$'s two children, $L$ and $L'$. The connections are made via a \emph{decision tree}~$\calT$, which we include in $\Pi'$. At a high level, the tree will satisfy the following.
\begin{enumerate}[label=(\arabic*)]
\item \emph{Root:} The root of the tree $\calT$ is identified with the vertex $T\cap R$ associated with $\rho$. That is, $\calT$ starts out with the bits in blocks $\fix(\rho)\subseteq[n]$ already queried.
\item \emph{Non-leaf:} The non-leaf vertices of $\calT$ query more bits, one block at a time.
\item \label{it:leaf}
\emph{Leaf:} Every leaf $\rho^*$ of $\calT$ extends some assignment $\tau$ that arises from the partitions of the children $L$, $L'$. Therefore, in $\Pi'$, we define $\rho^*$ to have $\tau$ as its unique child. (This way, the feasible sets satisfy $\Cube(\rho^*)\subseteq\Cube(\tau)$ as required in a decision-dag.)
\end{enumerate}

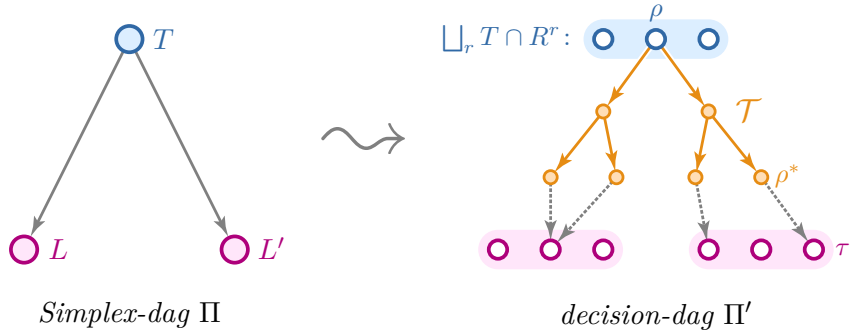
\begin{center}
\begin{lpic}[b(3mm)]{figs/dec-tree(.35)}
\lbl[c]{150,60;\scalebox{3}{\color{gray} $\leadsto$}}
\large
\lbl[c]{295,73;\color{myGold}$\calT$}
\normalsize
\lbl[c]{73,100;\color{myBlue}$T$}
\lbl[c]{33,20;\color{myPurple}$L$}
\lbl[c]{114,20;\color{myPurple}$L\smash{'}$}
\lbl[c]{60,-5;\itshape Simplex-dag $\Pi$}
\lbl[c]{260,-5;\itshape decision-dag $\Pi'$}
\small
\lbl[r]{231,100;\color{myBlue}$\bigsqcup_r T\cap R^r\colon$}
\lbl[c]{260,110;\color{myBlue}$\rho$}
\lbl[c]{310,48;\color{myGold}$\rho^*$}
\lbl[c]{331,20;\color{myPurple}$\tau$}
\end{lpic}
\end{center}

The tree $\calT$ is defined precisely as follows. Since $R^\circ\eqqcolon X\times\prod_{ij}Y^{ij}$ is $\rho$-structured, \autoref{lem:rho-like} produces an $x^*\in X$ such that $\{x^*\}\times\prod_{ij} Y^{ij}$ is $\rho$-like. Using the query alignment property for $L$ and $L'$, there are subsets $I,I'\subseteq[n]$, $|I\cup I'|\leq O(k/\log m)\leq O(d)$, such that any structured part in the partitions of $L$ and $L'$ that intersects the slice $\{x^*\}\times\prod_{ij}\calY^{ij}$ has their fixed blocks contained in~$I\cup I'$. We let $\calT$ query all bits in the blocks $(I\cup I')\smallsetminus \fix(\rho)$ in some order, and make the resulting vertices (having queried all bits in blocks $I\cup I'\cup \fix(\rho)$) the leaves of $\calT$.
\begin{claim}
Every leaf $\rho^*$ of $\calT$ satisfies \autoref{it:leaf}.
\end{claim}
\begin{proof}
Since $\rho^*$ extends $\rho$, and $\{x^*\}\times \prod_{ij} Y^{ij}$ is $\rho$-like, there is some $y^*\in \prod_{ij} Y^{ij}$ such that $\Ind_{\ell\times m}^n(x^*,y^*)\in \Cube(\rho^*)$. Since $(x^*,y^*)\in R^{\circ}\subseteq L\cup L'$, we have $(x^*,y^*)\in L$ or $(x^*,y^*)\in L'$. Suppose wlog that $(x^*,y^*)\in L$. Let $L\cap R'$, where $R'$ is $\tau$-structured, be the unique part of $L$ containing $(x^*,y^*)$. But since $\fix(\tau)\subseteq I \subseteq \fix(\rho^*)$ and both $\tau$ and $\rho^*$ agree with $\Ind_{\ell\times m}^n(x^*,y^*)$, we conclude that $\rho^*$ extends $\tau$, as required.
\end{proof}

\paragraph{Efficiency.} 
We remark that all the assignments appearing in $\Pi'$ have block-width $O(d)$. This holds for the vertices coming from partitioning of the simplicies in $\Pi$ due to our choice of $k\leq O(d\log m)$, and it holds for the vertices in the decision trees as they query at most $|I\cup I'|\leq O(d)$ additional blocks. This concludes the (simplified) proof of \autoref{thm:block-lift}.

\subsection{Accounting for error} \label{sec:errors}

Removing the assumption~\ref{assumption} is done virtually in the same way as in~\cite{garg18monotone}. We briefly recall the outline (and refer to \cite[\S5.3]{garg18monotone} for more details if necessary). Instead of partitioning each simplex independently from one another, we instead process them in reverse topological order, $T_1,\ldots,T_{m^d}$ (i.e., if $T_i$ is a descendant of $T_j$ then $i<j$), and before partitioning $T_i$ we first remove all error sets resulting from partitioning of its descendants. More precisely, we initialize an ``errorless'' box $B\coloneqq \calX\times\prod_{ij}\calY^{ij}$ and then process the simplicies as follows.

\medskip\noindent
{\slshape Iterate for $i=1,\ldots,m^d$:}
\vspace{-2mm}
\begin{enumerate}[label=(\arabic*)]
	\item Update $T_i$ by removing all the errors accumulated so far: $T_i~\leftarrow~T_i \cap B$. Note that $T_i$ continues to be a simplex.
	\item Apply \simlemma to obtain a box covering $\bigsqcup_r R^r\supseteq T_i$ with error sets $X^{\err}\subseteq\calX$, $Y^{\err,ij}\subseteq\calY^{ij}$. Output all the structured parts $R^r\cap T_i$ and discard the error parts.
	\item Update $B$ by removing all the error sets: $B\leftarrow B\smallsetminus (X^{\err}\times\prod_{ij} \calY^{ij} \cup \bigcup_{ij} \calX\times Y^{\err,ij}\times\prod_{i'j'\neq ij}\calY^{i'j'})$. Note that $B$ continues to be a box.
\end{enumerate}

We can now repeat the simplified proof of \autoref{sec:simple-proof} nearly verbatim \emph{using only the structured simplices output by the above process}. When processing $\Pi$'s root $T_{m^d}\cap B=B$, where $B\eqqcolon X\times\prod_{ij}Y^{ij}$ is the errorless box at the end of the process, we have that each of $X$, $Y^{ij}$ has density at least $1-m^d\cdot 2^{-k}=1-m^{-d}\geq 99\%$ by our choice of $k$. Hence $B$ is $*^n$-structured and we may assume that \simlemma produces the trivial partition for $B$. This yields the unique root for $\Pi'$ as before. Another key observation is that the associated errorless box $B$ \emph{grows} when we step from a simplex $T_i$ in $\Pi$ to either one of its children. Thus every structured part $R^r\cap T_i$ that is output by the above process is wholly covered by the structured parts of $T_i$'s children. This means that our discarding of error sets does not interfere with the construction of the internal trees in Step 2.


\section{Tree-like definitions} \label{sec:tree-like-definitions}

As in the dag-like case, we use the search-problem-centric view of proofs. Namely, recall from \autoref{sec:dag-definitions} that with any unsatisfiable $n$-variate CNF formula $F\coloneqq \land_{i\in[m]} D_i$ we associate a \emph{falsified-clause search problem} $S_F\subseteq\{0,1\}^n\times [m]$ where, on input a truth assignment $z \in\{0,1\}^n$, the goal is to find a clause $D_i$, $i\in[m]$, falsified by $z$, that is, $D_i(z)=0$.

Given a CNF search problem $S_F\subseteq\{0,1\}^n\times\calO$, we consider compositions with the (single-output) \emph{index gadget} $\Ind_m\colon[m]\times\{0,1\}^m\to\{0,1\}$ defined by $\Ind_m(x,y)\coloneqq y_x$; this is simply the column-index gadget of \autoref{sec:blocks} specialized to $\ell=1$. The discussion in \autoref{sec:cnf-enc} implies that $S_F\circ\Ind^n_m$ is at least as hard as $S_{F'}$ where $F'\coloneqq F\circ \Ind_m^n$ admits an efficient encoding:

\begin{fact} \label{fact:ind}
For any $n$-variate $k$-CNF $F$ with $|F|$ clauses, $F\circ \Ind_m^n$ is constructible in time $|F|m^k$.
\end{fact}

\subsection{Tree-like dags/proofs}
The definitions of decision-dags and LTF-dags from \autoref{sec:dag-definitions} can be straightforwardly specialized to decision-trees and LTF-trees (underlying dag is a tree). We define for any $S\subseteq\{0,1\}^n\times\calO$,
\begin{align*}
\decTree(S)~&\coloneqq~\text{least \emph{size} of a decision-tree solving $S$},\\
\ltfTree(S)~&\coloneqq~\text{least \emph{size} of an LTF-tree solving $S$},\\
\d(S)~&\coloneqq~\text{least \emph{depth} of a decision-tree solving $S$}.
\end{align*}
We recover the usual proof systems by further specializing to $S=S_F$,
\setlength\tabcolsep{.3em}
\begin{align*}
\begin{tabular}{rrl}
$\resTree(F)~=$
&$\decTree(S_F)~=$
&least \emph{length} of a tree-like Resolution refutation of $F$,\\
$\cutTree(F)~=$
&$\ltfTree(S_F)~=$
&least \emph{length} of a tree-like Cutting Planes refutation of $F$,\\
$\d(F)~=$
&$\d(S_F)~=$
&least \emph{depth} of a tree-like Resolution refutation of $F$.
\end{tabular}
\end{align*}

\subsection{Real protocols}

Consider a bipartite search problem $S\subseteq \calX\times\calY\times\calO$. A \emph{real protocol} $\Pi$ (introduced in~\cite{Kraj98}) is a binary tree satisfying the following. Each non-leaf node $v$ is labeled with a \emph{(combinatorial) triangle} $T_v\subseteq \calX\times\calY$, that is, a two-dimensional simplex (equivalently, $T$ is a triangle if it can be written as $T=\{(x,y) \in \calX\times\calY : a_T(x)<b_T(y)\}$ for some labelings of the rows $a_T\colon \calX \rightarrow \R$ and columns $b_T\colon \calY \rightarrow \R$ by real numbers). Each leaf $u$ of $\Pi$ is labeled with a solution $o_u\in\calO$. The protocol $\Pi$ solves $S$ if, for any input $(x,y)\in\calX\times\calY$ the unique root-to-leaf path, generated by walking left at node $v$ if $(x,y)\in T_v$ (and right otherwise), terminates at a leaf $u$ with $o_u\in S(x,y)$. We define
\begin{align*}
\rcc(S)~&\coloneqq~\text{least \emph{depth} of a real protocol solving $S$.}
\end{align*}
Real protocols are an established method for proving length lower bounds for tree-like CP:
\begin{lemma}[Kraj{\'{\i}}{\v{c}}ek~\cite{Kraj98}] \label{thm:real-cc}
Let $S_F'\subseteq \{0,1\}^{n_1}\times\{0,1\}^{n_2}\times\calO$ be the search problem $S_F$ for an unsatisfiable CNF formula $F$ together with an arbitrary bipartition of its $n_1+n_2$ variables. Then
\[
\rcc(S_F')~\leq~O(\log\cutTree(F)).
\]
\end{lemma}

\section{Tree-like lifting} \label{sec:tree-lifting}

The purpose of this section is to prove our small gadget lifting theorem (\autoref{thm:small-lift}), which would complete the proof of our main tree-like result (\autoref{thm:main-tree}). We restate this lifting theorem using the search-problem-centric language. \autoref{thm:small-lift} is an immediate corollary (replace $m$ by $m^{1000}$).

\begin{theorem}[Small gadget lifting] \label{thm:small-lift-s}
Let $F$ be an unsatisfiable CNF formula. For every $m$,
\[
m^{\Theta(\min(\d(F), m^{1/1000}))}
~\leq~
\ltfTree(S_F\circ \Ind_m^n)
~\leq~
\resTree(S_F\circ \Ind_m^n)
~\leq~
m^{O(\d(F))}.
\]
\end{theorem}

\paragraph{Upper bound.}
The last inequality is the trivial part of \autoref{thm:small-lift-s}. Indeed, since $S_F$ admits a decision-tree of depth $\d(F)$, we have that $S_F\circ\Ind_m^n$ admits one of depth $\d(F)\d(\Ind_m)=\d(F)(\log m+1)$. Therefore $\resTree(F\circ\Ind_m^n)\leq \exp(O(\d(F\circ\Ind_m^n))) = m^{O(\d(F))}$.

\paragraph{Lower bound.}
The first inequality is the nontrivial part of \autoref{thm:small-lift-s}. We formulate a lifting theorem for real protocols, which implies the first inequality by virtue of \autoref{thm:real-cc}.
\begin{theorem}[Real lifting] \label{thm:real-lift}
Let $F$ be an unsatisfiable CNF formula. For every $m$,
\[
\rcc(S_F\circ \Ind_m)~\geq~\min(\d(F),m^{1/1000})\cdot \Omega(\log m).
\]
\end{theorem}

The rest of this section is dedicated to proving \autoref{thm:real-lift}, which would conclude the proof of \autoref{thm:small-lift-s} and thereby the proof of our main tree-like result. Our proof follows the general plan familiar from previous tree-like lifting theorems, especially the exposition in~\cite{goos15deterministic,goos17bpp}.

\subsection{Proof of real lifting} \label{sec:tree-proof}

To prove \autoref{thm:real-lift}, we start with a given real protocol $\Pi$ of depth $d\leq o(m^{1/1000}\log m)$ for the composed problem $S_F\circ\Ind^n_m$ and construct a decision-tree of depth $O(d/\log m)$ for $S_F$.

The decision-tree is naturally constructed by starting at the root of $\Pi$ and taking a walk down the protocol tree guided by occasional queries to the variables $z=(z_1,\ldots,z_n)$ of $S_F$. During the walk, we maintain a rectangle $R\subseteq[m]^n\times(\{0,1\}^m)^n$ consisting of inputs that reach the current node in the protocol tree. Our goal is to ensure that the image $\Ind_m^n(R)$ has some of its bits fixed according to the queries to $z$ made so far, and the remaining bits sufficiently unrestricted.  We start by formulating our main technical lemma, which handles a single step (aka \emph{round}) in this walk: How to update $R$ while controlling the gadget outputs.

For terminology, recall from \autoref{def:entropy} the notions of \emph{min-entropy} $\Hmin$, \emph{blockwise min-entropy}, \emph{deficiency} $\Dmin$, the notation $\X\in X$ for a set $X$, and the marginal distribution $\X_S$ supported on~$X_S$.

The purpose of the following \roundlemma is to start with an $\X$ of moderate blockwise min-entropy ($\geq 0.9\log m$) and, through fixing some more gadget output bits (namely, $I$), bump the blockwise min-entropy back up ($\geq 0.95\log m$).

\begin{round-lemma} \label{lem:round-lemma}
Let $R\coloneqq X\times Y \subseteq [m]^N\times(\{0,1\}^m)^N$ be a rectangle for some set $N$. Suppose $\X$ has blockwise min-entropy at least $0.9\log m$ and $\Hmin(\Y)\leq d\log m$. Then there exists $I \subseteq N$ (we write $\bar{I}\coloneqq N\smallsetminus I$) such that for all $z_I \in \{0,1\}^I$ there exists a subrectangle $R'\coloneqq X'\times Y'\subseteq R$ s.t.
\begin{enumerate}[label={\itshape (\alph*)}]
\item $X'_I$ and $Y'_I$ are fixed (singletons) to produce a gadget output $\Ind_m^I(X'_I,Y'_I) = \{z_I\}$. \label{it:round-1}
\item $\X'_{\bar{I}}$ has blockwise min-entropy at least $0.95\log m$. \label{it:round-2}
\item $\Dmin(\X'_{\bar{I}})~ \leq~ \Dmin(\X) - \Omega(|I| \log m) + O(1)$. \label{it:round-3}
\item $\Dmin(\Y'_{\bar{I}})~ \leq~ \Dmin(\Y) + O(|I|)$. \label{it:round-4}
\end{enumerate}
\end{round-lemma}

\newcommand{\alg}{\hyperref[algorithm]{Decision-tree}\xspace}

\alg describes our query simulation of the real protocol $\Pi$. It repeatedly invokes the \roundlemma in each step down the protocol tree. Below, we analyze its correctness and efficiency.

\floatname{algorithm}{Decision-tree}
\begin{algorithm}[bht]
	\begin{algorithmic}[1]
		\vspace{1mm}
		\Input $z\in \set{0,1}^n$
		\Output solution to $S_F$
		\vspace{3mm}
		\State initialize $v=$ root of $\Pi$, $\rho=*^n$, $R\coloneqq [m]^n\times(\{0,1\}^m)^n$
  \While{$v$ is not a leaf}
  	\State let $T_v$ be the triangle associated with $v$
  	\State let $R'\coloneqq X'\times Y'\subseteq R$, $|R'|\geq |R|/4$, be such that $R'\subseteq T_v$ or $R'\cap T_v=\emptyset$; see \autoref{fig:triangle}
		\vspace{2mm}
    \State apply \roundlemma to $R'_{\free(\rho)}=X'_{\free(\rho)}\times Y'_{\free(\rho)}$ to obtain an $I\subseteq\free(\rho)$
    \State \textbf{Query} the variables $z_I\in\{0,1\}^I$
		\State let $R''\subseteq R'$ be such that $R''_{\free(\rho)}$ is the subrectangle given by \roundlemma for outputs $z_I$
		\vspace{2mm}
    \State update $R \leftarrow R''$ and $\rho_I \leftarrow z_I$
  	\State update $v$ to its left child if $R'\subseteq T_v$ and right child otherwise
   \EndWhile
  \State \textbf{Output} the same value as $v$ does
	\end{algorithmic}
	\caption{}
  \label{algorithm}
\end{algorithm}

\begin{figure}[t]
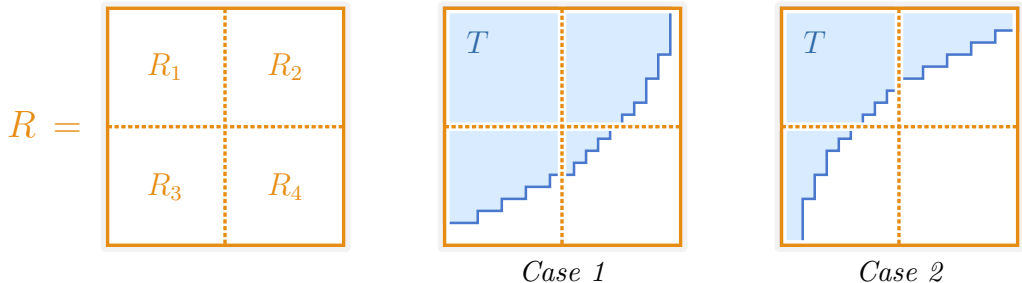

\begin{center}
\begin{lpic}[t(-3mm),b(-5mm)]{figs/quadrant(.32)}
\definecolor{myGold}{RGB}{231,141,20}
\definecolor{myBlue}{rgb}{0.19,0.41,.65}
\definecolor{myPurple}{RGB}{175,0,124}
\Large
\lbl[c]{-5,71;\color{myGold}$R\,=$}
\large
\lbl[c]{45,45;\color{myGold}$R_3$}
\lbl[c]{95,45;\color{myGold}$R_4$}
\lbl[c]{45,95;\color{myGold}$R_1$}
\lbl[c]{95,95;\color{myGold}$R_2$}

\lbl[c]{175,105;\color{myBlue}$T$}
\lbl[c]{315,105;\color{myBlue}$T$}

\normalsize
\lbl[c]{210,10;$\textit{Case 1}$}
\lbl[c]{350,10;$\textit{Case 2}$}

\end{lpic}
\end{center}
\caption{Simple fact: for every rectangle $R$ and triangle $T$ there is a subrectangle $R'\subseteq R$ with $|R'|\geq|R|/4$ that is either contained in $T$ or disjoint from it. Namely, after permuting the rows and columns of $R$ according to $T$'s orderings, take either the first quadrant $R_1$ or the fourth $R_4$.}%
\label{fig:triangle}%
\end{figure}

\begin{figure}[t]
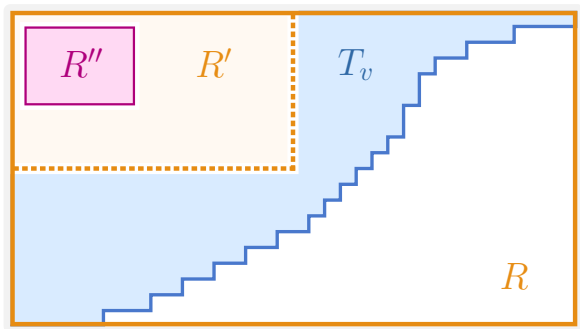

\begin{center}
\begin{lpic}[t(-3mm),b(-8mm)]{figs/tree-round(.42)}
\Large
\lbl[t]{43,108;\color{myPurple}$R''$}
\lbl[t]{85,108;\color{myGold}$R'$}
\lbl[t]{130,108;\color{myBlue}$T_v$}
\lbl[t]{180,40;\color{myGold}$R$}
\end{lpic}
\end{center}
\caption{A single iteration of \alg. The current rectangle $R$ is split by the triangle $T_v$. We choose $R'\subseteq R$ as a large rectangle either contained in $T_v$ or disjoint from it. Then we apply \roundlemma to $R'$, query relevant bits, and finally obtain $R''\subseteq R'$.}
\label{fig:tree-round}%
\end{figure}

\paragraph{Invariants and efficiency.}
We claim that $R=X\times Y$ and $\rho$ maintained by the decision-tree satisfy the following invariants at the start of the $i$-th iteration:
\begin{enumerate}[label={\itshape(\alph*')}]
\item[\emph{($\epsilon$')}] Write $I\coloneqq \fix(\rho)$ and $\bar{I}\coloneqq\free(\rho)=[n]\smallsetminus I$ for short. Then $|I|\leq O(i/\log m)$. \label{it:inv-1}
\item $X_I$ and $Y_I$ are fixed to produce a gadget output $\Ind_m^I(X_I,Y_I) = \{z_I\}$ (where $z_I=\rho_I$). \label{it:inv-2}
\item $\X_{\bar{I}}$ has blockwise min-entropy at least $0.95\log m$. \label{it:inv-3}
\item $\Dmin(\X_{\bar{I}})~ \leq~ O(i) - \Omega(|I| \log m)$. \label{it:inv-4}
\item $\Dmin(\Y_{\bar{I}})~ \leq~ O(i) + O(|I|)~\leq~O(i)$. \label{it:inv-5}
\end{enumerate}
The first item follows from \ref{it:inv-4} and the nonnegativity of deficiency. This shows that our decision-tree is efficient: it has depth $O(d/\log m)$, as desired. All the other properties are straightforward consequences of \ref{it:round-1}--\ref{it:round-4} in \roundlemma. We only need to check that the assumptions of \roundlemma are met every time it is invoked on line 5. The shrinking of $R$ down to $R'$ on line 4 can only lose at most 2 bits of min-entropy for the relevant random variables. Thus, if the blockwise min-entropy of $\X_{\free(\rho)}$ is $\geq 0.95\log m$ at the start of the iteration, then the blockwise min-entropy of $\X'_{\free(\rho)}$ is at least $\geq 0.95\log m-2\geq 0.9\log m$. Moreover, $\Pi$ has depth $d$, so \ref{it:inv-5} implies inductively that the precondition $\Hmin(\Y_{\bar{I}})\leq d\log m$ of \roundlemma is met.

\paragraph{Correctness of output.}
We finally have to argue that if we reach a leaf $v$ of $\Pi$, while maintaining $R$, $\rho$, then the solution output by $\Pi$ is also valid solution to $z$, of which the decision-tree knows that $z_{\fix(\rho)}=\rho_{\fix(\rho)}$. We need the following simple lemma (in fact, the \roundlemma will use a much stronger property of the gadgets, but we still give a short proof of the following).
\begin{lemma} \label{lem:fixing}
Consider any $R=X\times Y$ and associated $\rho$ during the execution of \alg. Then $Z\coloneqq \Ind_m^n(R)$ is such that $Z_{\fix(\rho)}=\{z_{\fix(\rho)}\}$ and $Z_i=\{0,1\}$ for every $i\in\free(\rho)$.
\end{lemma}
\begin{proof}
The claim about the coordinates $\fix(\rho)$ is \ref{it:inv-2}. Consider any $i\in\free(\rho)$. Then $\Hmin(\X_i)\geq 0.95\log m$ by \ref{it:inv-3}, and $\Dmin(\Y_i)\leq O(d)$ by \ref{it:inv-5}. Suppose for contradiction that $\Ind_m(\X_i,\Y_i)=0$, say. Hence all the $\geq m^{0.95}$ different values that $\X_i\in[m]$ can take are fixed to $0$ in $\Y_i\in\{0,1\}^m$. This implies $\Dmin(\Y_i)\geq m^{0.95}$. But this contradicts $\Dmin(\Y_i)\leq O(d)\leq m^{1/1000}\log m$.
\end{proof}

Suppose $\Pi$ outputs a clause $C$ of $F$ at the leaf $v$. Our goal is to show that $C(z)=0$, that is, that $C$ is a valid solution for $z$. By definition of $S_F\circ \Ind_m^n$ this means that $C(z')=0$ for all $z'\in \Ind_m^n(R)$, that is, all variables appearing in $C$ are fixed in the set $\Ind_m^n(R)$. By \autoref{lem:fixing} we must have that the variables of $C$ are contained in $\fix(\rho)$. Thus $C(z)=C(z_{\fix(\rho)})=0$, as desired.

This completes the proof of the real lifting theorem, assuming the \roundlemma.


\subsection{Overview of Round Lemma} \label{sec:round-lem}

To prove the \roundlemma we follow the general approach of \cite{goos17bpp,garg18monotone}, which we recap now along with what is needed to make it work for small gadgets where $m \ll N$.

Since the goal is to bump up the blockwise min-entropy from $0.9 \log m$ to $0.95 \log m$, we start by computing a blockwise min-entropy restoring partition of $X$, which simply takes a maximal assignment
that violates $0.95 \log m$ blockwise min-entropy, makes a part with all $x$'s that have that assignment, 
and then repeats on the rest of $X$ until all $x$'s are covered. The construction will guarantee
each part in the partition will fulfill all requirements for $X$, and so then we turn our attention to finding
a part with a fixed assignment $(I, \alpha)$ such that $Y$ is roughly uniform on the locations pointed
to by $(I, \alpha)$.

In \cite{garg18monotone}, the simplest way to prove this is to show that there is some $x \in X$
such that $Y$ is roughly uniform on {\it all} locations pointed to by $x$, and then simply take the rectangle part containing $x$. Because $X$ has high blockwise min-entropy
and $Y$ has very low deficiency, a Fourier argument
directly shows that for every set $I$, the expected parity of $\PMInd_m^I({\bf x}, {\bf y})$
is close to 0, where $\PMInd_m$ is the parity analogue of $\Ind_m$.
Taking a union bound over all such sets $I$, with high probability over $x \in X$
the expected parity of $\PMInd_m^I(x, {\bf y})$ is close to 0 for {\it all} sets $I$, which is equivalent (see e.g. \cite{goos17bpp})
to $\Ind_m^I(x, {\bf y})$ being close to uniform with high probability over $x$, and
choosing any such $x$ completes the lemma as stated before.

However, this union bound over all sets $I \subseteq [N]$ only works when ``close to 0'' is $N^{-\Omega(|I|)}$,
since there will be $N^{|I|}$ sets of that size. In reality the argument only shows 
$\Exp[\PMInd_m^I({\bf x}, {\bf y})] \leq m^{-\Omega(|I|)}$, which fails
in our case where $m \ll N$. Thus instead of disregarding the fixed $(I,\alpha)$ assignments of the rectangle partition
when looking at $x$, we will use the fact that the $(I,\alpha)$'s are the only coordinates of $y$
we care. While we have no control over the number of parts in the rectangle partition,
we {\it can} say that each individual part corresponds to an assignment of at most $O(d)$ coordinates.

We group $[N]$ into $\poly(d)$ ``megacoordinates'' of size $N/\poly(d)$ such that most $(I, \alpha)$ assignments in the rectangle partition
each only point to one value per megacoordinate 
(see \autoref{fig:rect-part} for an illustration).
We use the $(I,\alpha)$'s to replace the $x$s with shorter $x'$ vectors which only point to one value per megacoordinate, 
and repeat the argument in \cite{garg18monotone} but only using sets of
megacoordinates $I \subseteq [\poly(d)]$. Since $m = \poly(d)$, $m^{-\Omega(|I|)}$ is enough to cancel out 
$(\poly(d))^{|I|}$, and so the union bound goes through, giving an $x'$ that makes $\Ind_{m \cdot N/\poly(d)}^I(x', {\bf y})$ close to uniform.
Using the way we constructed the $x'$s out of the rectangle partition, this will give us an assignment $(I,\alpha)$
which is equally close to uniform from the partition, which completes the lemma.

\begin{figure}[t]
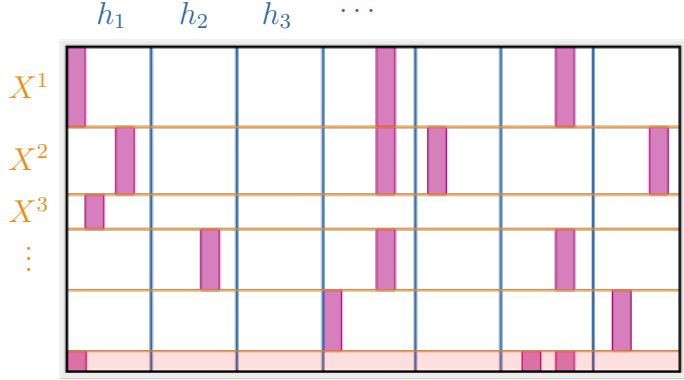

\begin{center}
\begin{lpic}[t(-2mm),b(-2mm)]{figs/partition(.5)}
\large
\lbl[t]{7,97;\color{myGold}$X^1$}
\lbl[t]{7,78;\color{myGold}$X^2$}
\lbl[t]{7,64;\color{myGold}$X^3$}
\lbl[t]{7,55;\color{myGold}\vdots}

\lbl[t]{29,116;\color{myBlue}$h_1$}
\lbl[t]{51,116;\color{myBlue}$h_2$}
\lbl[t]{73,116;\color{myBlue}$h_3$}
\lbl[t]{95,114;\color{myBlue}\ldots}
\end{lpic}
\end{center}
\caption{After partitioning $X$ into $\{X^j\}$ (purple regions are the coordinates of $I_j$, the restriction $\alpha_j$ to $I_j$ not pictured), we randomly block up the coordinate space $[N]$ into $\poly(d)$ megacoordinates (labeled $h_i$ here). With high probability only a small fraction of $X$ will be lost due to collisions.}
\label{fig:rect-part}%
\end{figure}



\subsection{Proof of Round Lemma}

We begin by performing a blockwise min-entropy restoring partition \cite{goos17bpp} on $X$.
\begin{itemize}
\item Initialize $\calF=\emptyset$. Iterate the following for $j = 1,2,\ldots$ until $X = \emptyset$:
\begin{itemize}
\item Let $I_j$ be a maximal (possibly empty) subset of $[N]$ such that 
$\X$ violates $0.95 \log m$-blockwise min-entropy on $I_j$, 
and let $\alpha_j \in [m]^{I_j}$ be an outcome witnessing this:
\[
\Pr[\X_{I_j}=\alpha_j]>2^{-0.95\log m}.
\]
\item Update ${\cal F} \leftarrow {\cal F} \cup \{(I_j, \alpha_j)\}$.
\item Define $X^j := \{x \in X:x_{I_j} = \alpha_j\}$ and update $X \leftarrow X \smallsetminus X^j$.
\end{itemize}
\item Return $X=\bigsqcup_j X^j$ and ${\cal F}$
\end{itemize}

Suppose the procedure returns $X=\bigsqcup_j X^j$ with associated ${\cal F} = \{(I_1, \alpha_1) \ldots (I_t, \alpha_t)\}$. For every $x \in X$ let $j(x)$ be the $j \in [t]$ such that $x \in X^j$.
By Lemma 5 of \cite{goos17bpp} it holds that $\Dmin(X^j_{I_j}) \leq \Dmin(\X) - 0.1|I_j|\log m + O(1)$ for all parts $X^j$ (except for some tiny parts, output late in the partitioning process, whose union covers at most 1\% of $X$; we tacitly ignore these parts).

For convenience, we assume that $N \geq d^5$. (Indeed, real lifting theorems already exist for large enough gadgets as discussed in \autoref{sec:tree-overview}; moreover, \roundlemma is only easier to prove in the regime of large $d$ and $m$.) We group the coordinates in $[N]$ into $d^3$ mega-coordinates.
Let ${\bf h}$ be a random variable which is uniform over all functions $h$ mapping $[N]\to[d^3]$ where $|h^{-1}(i^h)| = \frac{N}{d^3}$ for all $i^h \in [d^3]$.
Consider the subset of ${\cal F}$
consisting only of pairs $(I_j, \alpha_j)$ such that all coordinates in $I_j$ are
mapped to different mega-coordinates by $h$, or formally
$${\cal F}^h = \{(I_j, \alpha_j) \in {\cal F} : \forall i \neq i' \in I_j, h(i) \neq h(i')\}$$
Let $X_h \subseteq X$ be the union of all $X^j$ sets of the rectangle partition
such that $(I_j, \alpha_j) \in {\cal F}_h$.

\begin{claim}
With high probability over $h \sim {\bf h}$, we have $|X_h| \geq 0.99|X|$.
\end{claim}
\begin{proof}
We show that for a uniform choice of $x$ from $X$, with high probability
the unique part $X^{j(x)}$ which contains $x$ survives into $X_h$. 
See Figure~\ref{fig:rect-part} for an illustration.
Formally, $\Pr_{h \sim {\bf h}} [\Pr_{x \sim {\bf x}} (X^{j(x)} \not\subseteq X_h)] < 0.01$. First we consider the case of a fixed $x$. We will switch the calculation by treating $h$ as a fixed partition from ${\bf h}$ 
and treating $I_{j(x)}$ as a random set of size at most $10d$. To see that these are equivalent,
we can treat $h \sim {\bf h}$ as simply being a uniformly random permutation on $[N]$
with a fixed partition into $d^3$ equal sized megacoordinates, and so we can view 
$I_{j(x)}$ as a random set over $h([N])$.

Recalling that $N \geq d^5$, a straightforward calculation shows that
$$\begin{array}{rcl}
\Pr_{I_{j(x)}}(\forall i \neq i' \in I_{j(x)}: h(i) \neq h(i')) & = & \displaystyle \prod_{i = 0}^{10d} 1 - \frac{i \cdot (N/d^3 - 1)}{N - i} \\
& \geq & \displaystyle(1 - \frac{10d \cdot N/d^3}{N/2})^{10d} \\
& \geq & \displaystyle(1 - \frac{20}{d^2})^{10d} \\
& \geq & \displaystyle e^{-200/d} \geq 0.99
\end{array}$$
and so the same holds for $\Pr_{h \sim {\bf h}}(\forall i \neq i' \in I_{j(x)}: h(i) \neq h(i'))$
by our previous argument. Therefore
$$\begin{array}{rcl} \displaystyle\Pr_{h \sim {\bf h}} [\Pr_{x \sim {\bf x}} (X^{j(x)} \not\subseteq X_h)]
& = & \displaystyle\Pr_{h \sim {\bf h}} [\Pr_{x \sim {\bf x}} (\exists i \neq i' \in I_{j(x)}: h(i) = h(i'))] \\
& = & \displaystyle\Pr_{x \sim {\bf x}} [\Pr_{h \sim {\bf h}} (\exists i \neq i' \in I_{j(x)}: h(i) = h(i'))] \\
& \leq & \displaystyle\sum_{x \in X} \Pr_{x' \sim {\bf x}}(x' = x) \Pr_{h \sim {\bf h}}(\exists i \neq i' \in I_{j(x)}: h(i) = h(i')) \\
& < & \displaystyle\sum_{x \in X} \Pr_{x' \sim {\bf x}}(x' = x) \cdot 0.01 \\
& = & 0.01 \displaystyle\sum_{x \in X} \Pr_{x' \sim {\bf x}}(x' = x) = 0.01 \end{array}$$
which completes our claim.
\end{proof}

Henceforth, fix any $h$ satisfying $|X_h| \geq 0.99|X|$.
We shift to viewing each $y \in Y$ as a matrix $y^h \in Y^h$ with $m \cdot N/d^3$ rows and
$d^3$ columns in the canonical way, where each entry $((\alpha,i), i^h)$ in $y^h$ corresponds 
to the entry $(\alpha,i')$ in the original matrix $y$, where $i'$ is the $i$th element of the megacoordinate $i^h$.
Following our usual conventions let ${\bf x}_h$ be the uniform random variable
for selecting $x$ from $X_h$ and let and ${\bf y}^h$ be the uniform random variable
for selecting $y$ from $Y$ and viewing it as $y^h$ as described above.

Recall that $X$ satisfied $0.9 \log m$-blockwise min-entropy, and so for any $I \subseteq [N]$,
$H_{\infty}(X_I) \geq 0.9 \cdot |I| \log m$. Thus for all assignments $\alpha_I$,
$$\begin{array}{rcl} \displaystyle\Pr_{x \sim {\bf x}_h}(x_I = \alpha_I) 
& \leq & \displaystyle\frac{|X|}{|X_h|}\Pr_{x \sim {\bf x}}(x_I = \alpha_I) \\
& \leq & \frac{1}{0.99} \cdot 2^{-0.9 |I| \log m}  \leq 2^{-0.89 |I| \log m}
\end{array}$$
and so $X_h$ satisfies $0.89 \log m$-blockwise min-entropy.

Now we define the random variable ${\bm \alpha^h}$ on $([m] \times [\frac{N}{d^3}])^{d^3}$ 
to be a random restriction on $x$ that picks one location in each mega-coordinate and
assigns it a restriction $\alpha$. Note that this can also be viewed as choosing a location
in each column of $y^h$. The restriction will be sampled according to ${\cal F}^h$,
by first sampling $x \sim {\bf x}_h$ and taking all assignments in the corresponding pair $(I_j, \alpha_j)_{h(I_j)}$ where $j = j(x)$, and then choosing a random assignment $(i,\alpha^i)_{i^h}$ for all mega-coordinates $i^h$ left unassigned by $\alpha_j$.

\begin{figure}[H]
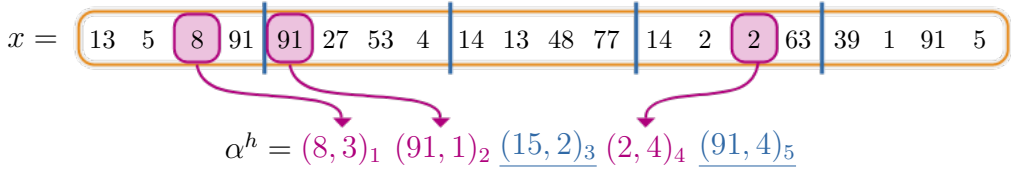

\begin{center}
\begin{lpic}[b(-7mm),t(-5mm)]{figs/mega-x(.50)}
\small
\lbl[l]{21,48;$13$}
\lbl[l]{35,48;$5$}
\lbl[l]{48,48;$8$}
\lbl[l]{58,48;$91$}
\lbl[l]{71,48;$91$}
\lbl[l]{83,48;$27$}
\lbl[l]{95,48;$53$}
\lbl[l]{108,48;$4$}
\lbl[l]{119,48;$14$}
\lbl[l]{131,48;$13$}
\lbl[l]{143,48;$48$}
\lbl[l]{155,48;$77$}
\lbl[l]{169,48;$14$}
\lbl[l]{183,48;$2$}
\lbl[l]{196,48;$2$}
\lbl[l]{206,48;$63$}
\lbl[l]{219,48;$39$}
\lbl[l]{232,48;$1$}
\lbl[l]{242,48;$91$}
\lbl[l]{256,48;$5$}

\large
\lbl[c]{88,19;\color{myPurple}$(8,3)_1$}
\lbl[c]{115,19;\color{myPurple}$(91,1)_2$}
\lbl[c]{143,19;\color{myBlue}\underline{$(15,2)_3$}}
\lbl[c]{169,19;\color{myPurple}$(2,4)_4$}
\lbl[c]{196,19;\color{myBlue}\underline{$(91,4)_5$}}

\lbl[c]{6,48;$x =$}
\lbl[c]{66,20;$\alpha^h =$}

\end{lpic}
\caption{Example of sampling $\alpha^h$ for $d^3 = 5$ megacoordinates of size $mN/d^3 = 4$.
Here $(I_j, \alpha_j)$ for $I_j = \{3, 5, 15\}$ and $\alpha_j = \{(8)_3, (91)_5, (2)_{15}\}$ is sampled. $(8)_3$ goes to $(8,3)$ in the first coordinate, $(91)_5$ goes to $(91,1)$ in the second
coordinate, and $(2)_{15}$ goes to $(2,4)$ in the fourth coordinate. For the third and fifth coordinate a pair in $[m] \times [5]$ is chosen uniformly, choosing $(15,2)$ for the third
and $(91,4)$ for the fifth.}
\end{center}
\label{fig:mega-sample}%
\end{figure}

Formally we define ${\bm \alpha^h}$ by the following procedure:
\begin{itemize}
\item sample $x \sim {\bf x}_h$ and let $j = j(x)$
\item for each $i^h \in h(I_j)$ let $i$ be the coordinate in $I_j$ mapping to $i^h$

and set $\alpha^h \leftarrow \alpha^h \cup ((\alpha_j)_i, i)_{i^h}$ 
\item for each $i^h \notin h(I_j)$ choose $i$ uniformly from $h^{-1}(i^h)$, 
choose $\alpha^i$ uniformly from $[m]$, and set $\alpha^h \leftarrow \alpha^h \cup (\alpha^i,i)_{i^h}$
\item return $\alpha^h$
\end{itemize}

Note that extending $\alpha_j$ uniformly to $\alpha^h$ does not change the min-entropy.
Thus because $X_h$ has blockwise min-entropy at least $0.89 \log m$, 
${\bm \alpha^h}$ has blockwise min-entropy at least $0.89 \log m$ as well, 
and the coordinates of every $\alpha^h$ are exactly $[d^3]$.

To proceed we now state a key lemma which is a generalized version of the Uniform Marginals Lemma of \cite{goos17bpp}. For completeness, we prove it in \autoref{sec:ind-lem}.


\begin{definition}[Multiplicative uniformity]
We say a random variable $\x\in S$ is $\epsilon$-multiplicatively uniform if $\Pr[\x=x]=(1 \pm \epsilon) \cdot \frac{1}{|S|}$ for all outcomes $x\in S$.
\end{definition}

\begin{index-lemma} \label{lem:index-lemma}
Let $\x \subseteq [\ell]^k$ and $\y \in (\{0,1\}^\ell)^k$ be random variables such that $\x$ has blockwise min-entropy $\geq 50\log k$ and $\Dmin(\y)\leq k$. Then there exists $x\in\operatorname{supp}(\x)$ such that $\Ind_{\ell}^k(x,\y)$
is $o(1)$-multiplicatively uniform.
\end{index-lemma}


We apply \indexlemma with $\x \coloneqq {\bm \alpha^h}$, $\y\coloneqq \Y^h$, $\ell \coloneqq mN/d^3$, $k\coloneqq d^3$. Note that $\Dmin(\y) \leq O(d) \leq k$ and that $\x$ has blockwise min-entropy $\geq 0.89\log m \geq 0.89\log d^{999} \geq 50\log d^3 = 50\log k$. We conclude that there is an $\alpha^h\in\operatorname{supp}{\bm \alpha^h}$ such that $\Ind_{mN/d^3}(\alpha^h, \y^h)$ is $o(1)$-multiplicatively uniform. Fix such an $\alpha^h$ and let $(I_j, \alpha_j)$ be any pair from which $\alpha^h$ can be sampled in our previous procedure.

We can now undo our grouping into mega-coordinates: Because $\Ind_{mN/d^3}(\alpha^h, {\bf y}^h)$ is $o(1)$-multiplicatively uniform, by marginalizing to $I_j$ we have that for all $x \in X^j$, $\Ind_m^{I_j}(x, {\bf y}) = \Ind_{mN/d^3}^{I_j}(\alpha_j, {\bf y}^h)$ is also $o(1)$-multiplicatively-close to uniform.

We now proceed to prove the lemma for $I \coloneqq I_j$. Hence let $z_I \in \{0,1\}^I$. We first take $X' = X^j$. For $Y'$ we first define $Y^{I,z_I} = \{y \in Y: \Ind_m^I(\alpha_j, y) = \{z_I\}\}$. We need to fix the rest of $Y_I^{I,z_I}$ in order for $Y_I'$ to be fixed, and so for $a \in \{0,1\}^{m|I|}$ we take $Y^a = \{y \in Y^{I,z_I}: y_I = a\}$.
Finally we let $Y' = Y^{\arg \max_a |Y^a|}$, or in other words we choose the largest $Y^a$
(obviously $Y^a$ is empty if $a$ is not consistent with $z_I$ in $\alpha_I$ so we can assume otherwise). We verify the properties \ref{it:round-1}--\ref{it:round-4}.

\begin{itemize}
\item[(a)] {\it $X_I'$ and $Y_I'$ are fixed and $\Ind_m^I(\alpha_j, y) = \{z_I\}$ for all $(x,y) \in R'$} \\
By definition of $X^j$, $x_{I_j} = \alpha_j$ for all $x \in X_I'$, and for $a$ such that $Y' = Y^a$ we have that $y_I = a$ for all $y \in Y' = Y^a$.
Note that $Y^a \subseteq Y^{I,z_I}$, and so by definition of $Y^{I,z_I}$ we know $y$ is fixed to $\{z_I\}$ on $\alpha_j$.
\item[(b)] {\it ${\bf X}_{\overline{I}}'$ has blockwise min-entropy at least $0.95 \log m$} \\
By the fact that all $I_j$ are maximally chosen in the rectangle partition, 
all rectangles $X^j$ have blockwise min-entropy $0.95 \log m$ on the coordinates $[N]-I_j$.
\item[(c)] {\it $\Dmin(\X'_{\bar{I}})~ \leq~ \Dmin(\X) - \Omega(|I| \log m) + O(1)$} \\
By Lemma 5 of \cite{goos17bpp}, each $X^j$ in the rectangle partition
has deficiency $\Dmin(\X) - 0.1|I|\log m + O(1)$.
\item[(d)] {\it $\Dmin(\Y'_{\bar{I}})~ \leq~ \Dmin(\Y) + O(|I|)$} \\
By the fact that $X^j \times Y$ is $o(1)$-multiplicatively-close to uniform and
by our definition of multiplicative uniformity,
$$\displaystyle\Pr_{y \sim {\bf y}}(y \in Y^{I,z_I}) \geq (1\pm o(1))2^{-|I|} \geq \frac{1}{2} \cdot 2^{-|I|}$$
and so $\Dmin(\Y_{\overline{I}}^{I,z_I})\leq \Dmin(\Y) + |I| + 1$. To move to $\Y_{\overline{I}}'$
we simply note that we chose the assignment $a$ that maximizes
$\Pr(\Y_I^{I,z_I} = a)$, which cannot increase $\Dmin(\Y_{\overline{I}}^{I,z_I})$.
\end{itemize}


\appendix

\section{Proof of \texorpdfstring{\indexlemma}{Large Index Lemma}}\label{sec:ind-lem}

We state two key lemmas before proving \indexlemma.
For convenience we shorten the base of the expectation when the variable 
in the inner expression is clear.
The first lemma is a standard application of Fourier analysis which appears
in different forms in many papers; we state the version needed to prove \indexlemma
and prove it at the end of this subsection, following the proof of \cite{vidick}.

\begin{lemma}\label{lem:fourier}
Let $\Lambda$ and $\Gamma$ be random variables on $X := [\ell]^k$ and $Y := (\{\pm 1\}^{\ell})^k$ respectively.
Assume that $\Lambda$ has blockwise min-entropy $\beta > 1/2$ and $\Gamma$ has deficiency $s$. 
Then for every $I \subseteq [k]$,
$$|\Exp_{\Lambda, \Gamma}[\chi_I(y_x)]| \leq (2^{-\beta/2 - 1}(k+s))^{|I|}$$
where $\chi_I(y_x) = \prod_{i \in I} y_i(x_i)$
\end{lemma}

The second lemma appeared in a different form in \cite{goos17bpp} as Lemma 9. We
omit the proof and defer interested readers to \cite{goos17bpp}.

\begin{lemma}\label{lem:uniform}
Let $x \in [\ell]^{k}$ and $Y \subseteq \{\pm 1\}^{\ell \times k}$ be such that 
$$|\Exp_{{\bf y}}[\chi_I(y_x)]| \leq 2^{-10|I|\log k}$$
for all $I \subseteq [k]$.
Then ${\bf y}_x$ is $1/k^3$-multiplicatively-close to uniform.
\end{lemma}

\begin{proof}[Proof of Large Index Lemma]
We map all $y$ from elements of $\{0,1\}^{\ell \times k}$ to
$(\{\pm 1\}^{\ell})^k$ in the natural way.
Applying \autoref{lem:fourier} we get that for all $I \subseteq [k]$
$$|\Exp_{\Lambda, {\bf y}}[\chi_I(y_x)]| \leq (2^{-25 \log k - 1}(k+k))^{|I|} \leq 2^{-20|I| \log k}$$
where the second inequality is by assumption. By Markov's inequality then, for any $I \subseteq [k]$
$$\displaystyle\Pr_{x \sim \Lambda}(|\Exp_{{\bf y}}[\chi_I(y_x)]| > 2^{-10|I|\log k}) \leq 2^{-10|I|\log k}$$
We say $x$ is {\it good} if $|\Exp_{{\bf y}}[\chi_I(y_x)]| \leq 2^{-10|I|\log k}$
for all $I \subseteq [k]$.
Taking a union bound over all such $I$ we get
$$\begin{array}{rcl} \displaystyle\Pr_{x \sim \Lambda}(x \mbox{ is not good}) 
& \leq & \displaystyle\sum_{I \subseteq [k]} \displaystyle\Pr_{x \sim \Lambda}(|\displaystyle \Exp_{{\bf y}}[\chi_I(y_x)]| > 2^{-10|I|\log k}) \\
& \leq & \displaystyle\sum_{I \subseteq [k]} 2^{-10|I|\log k} \\
& \leq & \sum_{t=1}^{k} \binom{k}{t}  2^{-10 t \log k} \\
& \leq & \displaystyle\sum_{t=1}^{k} 2^{-9 t \log k} \leq 2/k^9 \end{array}$$

Hence most $x$ are good, and by Lemma \ref{lem:uniform} for any good $x$
we have that $\Ind_k(x, {\bf y})$ is $1/k^3$-multiplicatively-close to uniform.
\end{proof}


\begin{proof}[Proof of \autoref{lem:fourier}]
Because marginalizing $\Gamma$ to any $S \subseteq \ell \times k$ cannot increase the deficiency of $\Gamma_S$ in $Y_S$,
it is enough to show that
$$|\Exp_{\Lambda, \Gamma}[\chi(y_x)]| \leq (2^{-\beta/2 - 1}(k+s))^k$$
Let $\Lambda(x) = \Pr(\Lambda = x)$. Because $\Lambda$ has blockwise min-entropy $\beta$,
it has Renyi entropy at least $\beta \cdot k$, meaning $\sum_x \Lambda(x)^2 \leq 2^{-\beta \cdot k}$.
By Cauchy-Schwarz
$$\begin{array}{rcl} \displaystyle|\Exp_{\Lambda, \Gamma} [\chi(y_x)]| 
& = & \displaystyle\sum_x \Lambda(x) |\Exp_{\Gamma}[\chi(y_x)]| \\
& \leq & (\displaystyle\sum_x \Lambda(x)^2)^{1/2} (\sum_x |\Exp_{\Gamma}[\chi(y_x)]|^2)^{1/2} \\
& \leq & 2^{-(\beta/2) k} \cdot (\displaystyle\sum_x |\Exp_{\Gamma}[\chi(y_x)]|^2)^{1/2} \\
& = & 2^{-(\beta/2) k} \cdot (\displaystyle\sum_x |\Exp_{\Gamma}[\chi(y_x)]|^2)^{1/2}
\end{array}$$
We thus turn our attention to proving a bound on $\sum_x |\Exp_{\Gamma}[\chi(y_x)]|^2$. 
Let $\chi_{\geq i}(y_x) = \chi_{\{i \ldots k\}}(y_x)$. Again by Cauchy-Schwarz
$$\begin{array}{rcl}
\displaystyle\sum_x |\Exp_{\Gamma}[\chi(y_x)]|^2  & = & \displaystyle\sum_x |\prod_i \Exp_{\Gamma}[\chi_{\geq i}(y_x)]|^2 \\
& \leq & \displaystyle\sum_x \prod_i \Exp_{\Gamma}[\chi_{\geq i}(y_x)]^2 \\
& = & \displaystyle\sum_{x_2 \ldots x_k} \prod_{i \geq 2} \Exp_{\Gamma}[\chi_{\geq i}(y_x)]^2 \cdot \sum_{x_1} \Exp_{\Gamma}[\chi_{\geq 1}(y_x)]^2
\end{array}$$
Since $\Hmin(\Gamma) \geq \ell k - s$, for a fixed $x_2 \ldots x_k$
$$\Ent(\chi_{\geq 1}(y_x)) = \Ent(\Gamma_1 \mid \Gamma(x_2) \ldots \Gamma(x_k)) \geq \ell - (k + s)$$
By Pinsker's inequality $\Exp_{\Gamma}[\chi_{\geq 1}(y_x)]^2 \leq (1-\Ent(\chi_{\geq 1}(y_x)))/2$, and so by sub-additivity of the expectation
$$\displaystyle\sum_{x_1} \Exp_{\Gamma}[\chi_{\geq 1}(y_x)]^2 \leq (\ell - (\ell - (k + s))/2 = (k + s)/2$$
Plugging this back into our previous expression we get
$$\displaystyle\sum_{x_2 \ldots x_k} \prod_{i \geq 2} \Exp_{\Gamma}[\chi_{\geq i}(y_x)]^2 \cdot \sum_{x_1} \Exp_{\Gamma}[\chi_{\geq 1}^2(y_x)]^2 \leq \frac{k+s}{2}\sum_{x_2 \ldots x_k} \prod_{i \geq 2} \Exp_{\Gamma}[\chi_{\geq i}(y_x)]^2$$
Finally we repeat for all $i = 2 \ldots k$, and in the end we get
$$\begin{array}{rcl}
\displaystyle\sum_x |\Exp_{\Gamma}[\chi(y_x)]|^2 
& \leq & \displaystyle\sum_x \prod_i \Exp_{\Gamma}[\chi_{\geq i}(y_x)]^2 \\
& \leq & \frac{k+s}{2}\displaystyle\sum_{x_2 \ldots x_k} \prod_{i \geq 2} \Exp_{\Gamma}[\chi_{\geq i}(y_x)]^2 \\
& & \ldots \\
& \leq & (\frac{k+s}{2})^k \displaystyle \prod_{i > k} \Exp_{\Gamma}[\chi_{\geq i}(y_x)]^2 = (\frac{k+s}{2})^k
\end{array}$$
Putting this bound on $\sum_x |\Exp_{\Gamma}[\chi(y_x)]|^2$ together with the earlier proof completes the lemma.
\end{proof}

\section{Proof of Simplex Lemma} \label{sec:proof-sim-lemma}

The proof of \simlemma is a small modification of the proof of the Triangle Lemma (the case of 2-dimensional simplices) in~\cite{garg18monotone}. Since the proof for the latter is somewhat long, we describe here only the required modifications. Our discussion naturally assumes familiarity with the original proof~\cite{garg18monotone}, which analyzed a partitioning procedure called Triangle Scheme (with subroutines Rectangle Scheme and Column Cleanup). The basic difference between the two settings is that instead of partitioning a 2-dimensional simplex over $\calX\times \calY$, Bob's input in $\calY$ is further shared over $n\ell$ many Bobs, that is, $\calY$ is replaced with $\prod_{ij}\calY^{ij}$. In this appendix, we explain how to replace all parts involving Bob with multi-party analogs. There are two: (1) Rectangle Scheme, and (2) Column Cleanup.

\paragraph{(1) Rectangle Scheme.}
Our first observation is that the Rectangle Scheme, which partitions rectangles $R\subseteq \calX\times\calY$, works equally well to partition boxes $B\subseteq \calX\times\prod_{ij}\calY$. Indeed, each part output by Rectangle Scheme is obtained from $R\coloneqq X\times Y$ by restricting the set $X$ arbitrarily and, crucially, restricting $Y$ only via bit-wise restrictions (Round 2 of Rectangle Scheme fixes pointed-to bits in all possible ways). But such bit-wise restrictions when applied to a box $B\coloneqq X\times \prod_{ij} Y^{ij}$ still result in a box. With this understanding, we may apply Rectangle Scheme to a box.

\paragraph{(2) Column Cleanup.}
Our biggest modification is to replace the Column Cleanup procedure with a natural multi-party analog. We start with a lemma saying that either a simplex over Bobs' domains contains a box that satisfies the largeness condition of $\rho$-structuredness (\autoref{def:structured}), or the simplex can be covered with a small error set.

\begin{claim}
Let $T\subseteq \prod_{ij} \calY^{ij}$ be a simplex. Then one of the following holds.
\begin{enumerate}[label=(\roman*)]
\item $T$ contains a box $B\coloneqq\prod_{ij} Y^{ij}$ where each $Y^{ij}$ has density $\geq \smash{2^{-m^{1/2}}}$ (i.e.,~$\Dmin(\bm{Y}^{ij})\leq m^{1/2}$).
\item $T$ is covered by $\bigcup_{ij} Y^{ij,\err}\times \prod_{i'j'\neq ij} \calY^{i'j'}$ where each $Y^{ij,\err}$ has density $\leq 2^{-m^{1/2}}$.
\end{enumerate}
\end{claim}
\begin{proof}
Consider the largest \emph{cube} $B\coloneqq \prod_{ij} Y^{ij}$ contained in $T$, that is, where all the sets $Y^{ij}\subseteq\calY^{ij}$ have the same size. The largest cube can be obtained by the following process: Identify each $\calY^{ij}=\{0,1\}^m$ with $[N]$ according to the reverse of the ordering given to $\calY^{ij}$ by $T$. (Thus if $x\in T\subseteq [N]^{n\ell}$ and $x'\leq x$ coordinate-wise then $x'\in T$.) Then $B$ equals $[M]^{n\ell}$ where $M$ is the largest number such that $(M,\ldots,M)\in T$. If some (and hence every) $Y^{ij}$ has density $\geq 2^{-m^{1/2}}$ we are in case \emph{(i)}. Otherwise we claim we are in case \emph{(ii)} with $Y^{ij,\err}\coloneqq Y^{ij}$. Indeed, consider any $x\coloneqq(M_{11},\ldots,M_{n\ell})\in T$. We must have $M_{i^*j^*}\leq M$ for some $(i^*,j^*)\in[n]\times[\ell]$ since otherwise by monotonicity $(M+1,\ldots,M+1)\in T$ contradicting our choice of $M$. But then $x \in Y^{i^*,j^*}\times\prod_{i'j'\neq i^*j^*}\calY^{i'j'}$, as required.
\end{proof}

We say that a simplex $T\subseteq\prod_{ij}\calY^{ij}$ is \emph{empty-or-heavy} iff $T=\emptyset$ or $T$ satisfies case \emph{(i)} above.

\floatname{algorithm}{Bob Cleanup}
\begin{algorithm}[ht]
	\begin{algorithmic}[1]
		\vspace{1mm}
		\Input Simplex $T \subseteq \calX \times \prod_{ij}\calY^{ij}$
		\Output Error sets $Y^{ij,\err} \subseteq \calY^{ij}$ and their combination $Y^{\err}$ 
		\vspace{3mm}
  \State initialize $Y^{ij,\err} \gets \emptyset$ and write $Y^{\err} \coloneqq \bigcup_{ij} Y^{ij,\err}\times\prod_{i'j'\neq ij}\calY^{i'j'}$ as a function of the $Y^{ij,\err}$
	\State For $I \subseteq [n]$, $\alpha \in [m]^I$, $\gamma \in (\{0,1\}^\ell)^I$, define $Y_{I,\alpha,\gamma} \coloneqq \big\{y \in \prod_{ij}\calY^{ij} : g^I(\alpha, y_I) = \gamma\big\}$
  \While{there are $I, \alpha, \gamma, x\in\calX$ s.t.\ $T'\coloneqq T \cap (\{x\} \times (Y_{I,\alpha,\gamma}\smallsetminus Y_{\err}))$ is not empty-or-heavy}
Add to the $Y^{ij,\err}$ all error sets from case \emph{(ii)} for $T'$
    \EndWhile
  \State \textbf{Output} $Y^{ij,\err}$ and $Y^{\err}$
	\end{algorithmic}
  \label{alg:bob}
	\caption{}
\end{algorithm}
\newcommand{\bobclean}{\hyperref[alg:bob]{Bob Cleanup}\xspace}

The below claim is the multi-party analog of Claim 10 in~\cite{garg18monotone}. This completes the modifications needed to the proof of the Triangle Lemma to handle multiple Bobs.

\begin{claim}\label{clm:col-clean}
For a simplex $T \subseteq \calX \times \prod_{ij}\calY^{ij}$, let $Y^{ij,\err}$, $Y^{\err}$ be the outputs of \bobclean. Then:
\begin{itemize}
	\item \emph{Empty-or-heavy:}~ For every triple $(I \subseteq [n], \alpha \in [m]^I, \gamma \in (\{0,1\}^\ell)^I)$, and every $x \in \calX$, it holds that $T \cap (\set{x} \times (Y_{I,\alpha,\gamma}\smallsetminus Y_{\err}))$ is empty-or-heavy.
	\item \emph{Size bound:}~ $|Y^{ij,\err}| \le 2^{m - \Omega(m^{1/2})}$ for every $i,j$.
\end{itemize}
\end{claim}
\begin{proof}
The first property is immediate by definition of \bobclean. For the second property, in each while-iteration, at most $2^{m-m^{1/2}}$ elements get added to each $Y^{ij,\err}$. Moreover, there are no more than $2^n \cdot m^n \cdot 2^{n\ell} \cdot m^n = (2m)^{2n\ell}$ choices of $I$, $\alpha$, $\gamma$, $x$, and the loop executes at most once for each choice. Thus, $|Y^{ij,\err}| \le (2m)^{2n\ell} \cdot 2^{m - m^{1/2}} \le 2^{m- \Omega(m^{1/2})}$.
\end{proof}

\bigskip\bigskip
\subsection*{Acknowledgements}
We thank Robert Robere for discussions and anonymous STOC reviewers for comments. The first author was supported in part by the NSF grant No.\ CCF-1412958. The third and fourth authors were supported by NSERC, and the fourth author was supported in part by NSF grant No.\ CCF-1900460.

\pagebreak

\DeclareUrlCommand{\Doi}{\urlstyle{sf}}
\renewcommand{\path}[1]{\small\Doi{#1}}
\renewcommand{\url}[1]{\href{#1}{\small\Doi{#1}}}
\bibliographystyle{alphaurl}
\bibliography{cp-aut}
\end{document}